%% file: ms.tex
\documentclass[onefignum,onetabnum]{siamart190516}

\input{shared}

\ifpdf
\hypersetup{
  pdftitle={Quantifying Uncertainty for Temporal Motif Estimation in Graph Streams under Sampling},
  pdfauthor={Xiaojing Zhu, and Eric D. Kolaczyk}
}
\fi

\begin{document}

\maketitle

\begin{abstract}
Dynamic networks, a.k.a. graph streams, consist of a set of vertices and a collection of timestamped interaction events (i.e., temporal edges) between vertices. Temporal motifs are defined as classes of (small) isomorphic induced subgraphs on graph streams, considering both edge ordering and duration. As with motifs in static networks, temporal motifs are the fundamental building blocks for temporal structures in dynamic networks. Several methods have been designed to count the occurrences of temporal motifs in graph streams, with recent work focusing on estimating the count under various sampling schemes along with concentration properties. However, little attention has been given to the problem of uncertainty quantification and the asymptotic statistical properties for such count estimators. In this work, we establish the consistency and the asymptotic normality of a certain Horvitz-Thompson type of estimator in an edge sampling framework for deterministic graph streams, which can be used to construct confidence intervals and conduct hypothesis testing for the temporal motif count under sampling. We also establish similar results under an analogous stochastic model. Our results are relevant to a wide range of applications in social, communication, biological, and brain networks, for tasks involving pattern discovery.  
\end{abstract}



\section{Introduction}
In the age of big data, streaming data has become ubiquitous in various modern applications, for example, e-commerce purchases, real-time surveillance, click-streams from websites, and players' activities in online gaming.
Consequently, there is a growing literature on streaming data analysis including clustering, classification, anomaly detection, pattern mining, etc. (see survey papers \cite{gomes2019machine, ramirez2017survey}).
Streaming data is a temporally ordered and potentially infinite sequence of objects that arrive continuously over time. A particularly important subclass of such data is the streaming relational data, which arrives not only continuously, but also carries the relational information about a pair of nodes interacting. 
Examples are data streams from 
electronic communication (e.g., phone calls, emails), Internet of things (e.g., routers, connected devices) and content sharing platforms (e.g., Facebook, Twitter) that record interactions between two individuals (or entities) with timestamps. 
Due to its relational nature, this type of streaming data has been closely linked with the literature on network analysis and graph algorithms, and referred to by a variety of names - dynamic/temporal networks, temporal graphs, edge/graph streams, etc. (we use these terms mutually interchangeablely in this paper). 

Leveraging concepts and tools in network analysis, much work has been done on the analysis of graph streams, including anomaly detection \cite{aggarwal2011outlier, eswaran2018sedanspot}, pattern matching \cite{ chen2017unified, sanei2019fleet, song2014event}, etc., among which an important effort revolves around developing scalable approaches for streams with increasingly large volume and high velocity. Sampling has become a central tool and various sampling methods were proposed under this context for a variety of tasks in temporal networks \cite{ahmed2021online, ahmed2017sampling, rocha2017sampling, shin2017wrs}. 
In this work, we study one such problem in this area, concerning the use of sampling to quantify the frequency of motifs in temporal networks. 
We provide uncertainty quantification in a certain sampling model by studying the asymptotic properties of the motif count estimator as the stream flows indefinitely. 
This is relevant to a wide range of applications in social networks \cite{kovanen2013temporal}, communication networks \cite{ tu2018network, zignani2017temporal}, biological networks \cite{chechik2008activity, prvzulj2007biological}, brain networks \cite{battiston2017multilayer, wei2017identifying}
etc., for pattern discovery (see \cite{liu2021temporal} for a recent survey).

\subsection{Related Work}
Here we give a summary of the current development on network motif counting and identify the research gap we aim to fill with our work. 

Recurring subgraph patterns in static networks, termed `motifs' \cite{milo2002network}, are considered basic structural elements in understanding networks and their underlying complex systems. Counting motifs has been an important statistical and computational problem in the analysis of static networks. In the context of counting motifs in large network graphs, enumerating all occurrences of a given motif could be prohibitively expensive. Therefore, network sampling and estimation are often used to provide an approximation to the true motif count in a computationally efficient manner. Much work has been done in this direction for estimating motif counts under various sampling schemes for static networks \cite{kashtan2004efficient, klusowski2018counting}. Very recently, Bhattacharya et al. (2020) \cite{bhattacharya2020motif} made an advance on the development of statistical inference framework for the motif estimation problem in static networks, where they provided results on the asymptotic properties of a motif count estimator in the subgraph sampling model.

On the other hand, as part of the still-emerging field of dynamic (or temporal) networks 
\cite{holme2015modern,  holme2019temporal}, several ways have been proposed to extend the notion of motifs to the context of graph streams, consisting of a set of vertices and a collection of time-stamped interaction events \cite{hulovatyy2015exploring, kovanen2011temporal, paranjape2017motifs, song2014event}. One widely used notion is from Paranjape et al. (2017) \cite{paranjape2017motifs}, where temporal motifs are defined as classes of isomorphic induced subgraphs on sequences of temporal edges, considering both edge ordering and duration. Similar to motifs in static networks, temporal motifs are also the fundamental building blocks for temporal structures in dynamic networks.

Several methods have been designed to count the occurrences of temporal motifs in graph streams (see \cite{jazayeri2020motif} for a survey), with recent work focusing on estimating the count under various sampling schemes along with establishing concentration properties \cite{liu2018sampling, wang2020efficient, sarpe2021presto}. However, little attention has been given to the problem of uncertainty quantification and the asymptotic statistical properties of these temporal motif count estimators. This paper aims to fill this gap by studying the asymptotics in motif estimation.

\subsection{Our Contributions and Paper Outline}
In this work, we consider the problem of estimating the temporal motif count (i.e., the number of temporal motifs in a graph stream) under an edge sampling model, 
where each edge is sampled independently with probability $p$ and the local motif count around the sampled edge is observed. 
We provide uncertainty quantification under this sampling model through a study of asymptotic properties.
Specifically, we establish the consistency and the asymptotic normality for a certain Horvitz-Thompson type estimator under sampling (proposed by Wang et al. (2020) \cite{wang2020efficient}) in deterministic graph streams, as the number of temporal edges grows indefinitely. We also establish similar results under an analogous stochastic model. These results can be used to construct confidence intervals and conduct hypothesis testing for the temporal motif count under this sampling model. While the proof of the asymptotic results for the deterministic case is relatively straightforward, the proof for the stochastic case is nontrivial. The challenge in the latter case arises from the need to analyze the behavior of a stream of dependent random variables that emerge in our problem as a result of applying a type of sliding window to an underlying marked Poisson point process. More details are provided in proofs in appendix.

The paper is organized as follows. Notation and background are in section \ref{sec:background}, our main results are in
section \ref{sec:main}, experimental
results are in section \ref{sec:numerical}, and some discussion follows in
section \ref{sec:discussion}.

\section{Background}
\label{sec:background}
In this section, we provide essential notation and background.
\subsection{Notation and Definitions}

We first provide formal definitions for temporal graphs and temporal motifs. 

\begin{definition}[Temporal graph] \label{def-graph}
A temporal graph $T_m = \{(u_i , v_i , t_i ) = (e_i, t_i)$, $i = 1, \cdots, m\}$ on node set $V_m$ is defined as a collection of timestamped directed edges $e_i$, a.k.a., temporal edges, where each $u_i$ and $v_i$ are elements of $V_m$ with $u_i \neq v_i$, each $t_i$ is a timestamp in $\mathbb{R}^+$, and $t_1< t_2, \cdots, <t_m$. 
\end{definition}

\begin{definition}[$\delta$-temporal motif \cite{paranjape2017motifs}]\label{def-motif}
A $k$-node, $l$-edge $\delta$-temporal motif 
$H = \{(u_1, v_1, t_1), (u_2, v_2, t_2), \cdots, (u_l, v_l, t_l)\}$ is a sequence of $l$ edges that are time-ordered within a $\delta$ duration, i.e., $t_1< t_2, \cdots, <t_l$ and $t_l-t_1 \leq \delta$, such that the induced static graph from the edges is connected and has $k$ nodes.
\end{definition}
\textit{Remark.} The induced static graph $G$ is obtained from temporal graph $T$ by ignoring all timestamps of edges, i.e., each edge $(u, v)$ in $G$ is associated with a temporal edge $(u, v, t)$ in $T$. Notice that in Definitions \ref{def-graph} and \ref{def-motif}, edges are allowed to be recurrent, meaning that the same edge can occur at different time points. Then the induced static graph from the defined temporal graph or temporal motif could be a multi-graph. 

In Definition \ref{def-motif}, a temporal motif provides a template for a particular pattern in a specified duration of time $\delta$. We are interested in counting the number of occurrences of such a pattern in a given temporal network. An occurrence of an $l$-edge $\delta$-temporal motif $H$ in a given temporal network $T_m$ is defined as $l$ temporal edges in $T_m$ satisfying three conditions: 1) the static multi-graph induced from the $l$ temporal edges is isomorphic to that induced from the temporal motif; 
2) the ordering of the matching edges are the same; 3) the $l$ temporal edges are within the duration time $\delta$. It is also referred to as a $\delta$-instance of motif $H$ by Liu et al.\cite{liu2018sampling}, the formal definition of which is provided below. Throughout the paper, we usually refer to a $\delta$-instance of motif $H$ simply as an instance of motif $H$ when no confusion is likely. 

\begin{definition}[motif $\delta$-instance\cite{liu2018sampling}]\label{def-instance}
A time-ordered sequence
$S = \{(w_1,x_1,t_1^\prime),\cdots, (w_l ,x_l ,t_l^\prime)\}$ of $l$ temporal edges from a given temporal graph $T_m$ is a $\delta$-instance of temporal motif $H =
\{(u_1,v_1,t_1),\cdots,(u_l ,v_l,t_l )\}$ if 1) there exists a bijection $f$ on the vertices such that $f(w_i)=u_i$, $f(x_i)=v_i$, $i=1, \cdots, l$; 2) the edges all occur within duration $\delta$, i.e., $t_l^\prime - t_1^\prime < \delta$. 

\end{definition}

\subsection{Temporal Motif Estimator under Edge Sampling Regime}
Suppose $T_m = \{(u_i , v_i , t_i ) = (e_i, t_i)$, $i = 1, \cdots, m\}$ is a temporal graph on node set $V_m$, $H$ is a $k$-node, $l$-edge $\delta$-temporal motif. Let $C(H, T_m)$ denote the number of instances of temporal motif $H$ in temporal graph $T_m$. The goal is to estimate $C(H, T_m)$ via some sampling regime. We now formally describe one of the state-of-the-art sampling methods for temporal motif estimation, the edge sampling regime\cite{wang2020efficient}. 

The exact count of temporal motif $H$ in $T_m$ can be written as
 \begin{equation}
 C(H, T_m) = \frac{1}{l} \sum_{i = 1}^m \eta(e_i),
 \end{equation}
where $\eta(e_i)$ is the number of instances of 
temporal motif $H$ in $T_m$ containing an edge $e_i$, which can be regarded as the local motif count. $\eta(e_i)$ can be calculated using a well-established backtracking algorithm\cite{mackey2018chronological} with time complexity
$O(l r^{l-1})$, 
where $r$ is the expected number of edges within time span $\delta$. To get the exact count of the number of instances of $H$ in $T_m$, we first obtain the local count $\eta(e_i)$ for each $e_i \in T_m$ and sum them up. Then the total number of instances of $H$ in $T_m$ can be obtained by dividing the sum by $l$ as each instance contains $l$ edges thus being counted $l$ times in obtaining the local counts for all edges in $T_m$. 

In an edge sampling regime proposed by \cite{wang2020efficient}, 
each temporal edge in $T_m$ is sampled independently with probability $p \in (0, 1)$. Let $\omega_i$ be the indicator of the event that edge $e_i$ is sampled. Then the estimator for $ C(H, T_m) $ is 
 \begin{equation}\label{estimator}
\hat C(H, T_m) = \frac{1}{p l } \sum_{i = 1}^m \omega_i\eta(e_i),
 \end{equation}
Note that $E[ \hat{C}(H, T_m)] = C(H, T_m)$, hence it is an unbiased estimator for the count $C(H, T_m)$. It is also a Horvitz-Thompson type of estimator as it uses inverse probability weighting to achieve unbiasedness.

\subsection{Asymptotic Notation}
We use the following standard notation for the asymptotic behavior of the relative order of magnitude of two sequences of numbers. For two positive sequences $\{a_n\}_{n\geq1}$ and $\{b_n\}_{n\geq1}$, 
\begin{itemize}
    \item $a_n = O(b_n)$ means $a_n \leq C_1b_n$ for all $n$ large enough and positive constant $C_1$.
    \item $a_n = \Omega(b_n)$ means $a_n \geq C_2b_n$ for all $n$ large enough and positive constant $C_2$.
    \item $a_n = \Theta(b_n)$ means $C_2b_n \leq a_n \leq C_1b_n$, for all $n$ large enough and positive constants $C_1$, $C_2$. 
    \item $a_n \asymp b_n$ if $a_n = \Theta(b_n)$. This is sometimes expressed by saying that $a_n$ and $b_n$ are of the same order of magnitude.  
    \item $a_n \lesssim b_n$ means $a_n = O(b_n)$, $a_n \gtrsim b_n$ means $a_n = \Omega(b_n)$.
    \item $a_n \gg b_n$ means $b_n = o(a_n)$, i.e. $\lim_{n\rightarrow \infty} b_n/a_n =0$.
\end{itemize}
When discussing asymptotics of random variables, we use the following standard notation. Let $X_n$ be random variables and $a_n$ positive real numbers, we define
\begin{itemize}
    \item $X_n = O_p(a_n)$ if for every $\delta >0$ there exist constants $C_\delta$ and $n_0$ such that $P(|X_n| \leq C_\delta a_n) > 1-\delta$ for every $n \geq n_0$. This is also expressed by saying that the sequence $X_n/a_n$ is bounded in probability. 
    \item $X_n = o_p(a_n)$ if 
    for every $\epsilon >0$, $ P(|\frac{X_n}{a_n}| > \epsilon)  \rightarrow 0$, as $n \rightarrow \infty$. This is also expressed by saying that the sequence $X_n/a_n$ converges in probability to zero. 
\end{itemize}
Two convergence concepts for random variables are used in this paper. 
\begin{itemize}
    \item $X_n$ converges in probability to the random variable $X$ as $n\rightarrow \infty$, shown by $X_n \xrightarrow{p} X$, as $n\rightarrow \infty$, if $\forall \epsilon>0$, $P(|X_n-X|>\epsilon) \rightarrow 0$, as $n\rightarrow \infty$.
    \item $X_n$ converges in distribution to the random variable $X$ as $n\rightarrow \infty$, shown by $X_n \xrightarrow{d} X$, as $n\rightarrow \infty$, if $P(X_n \leq x) \rightarrow P(X\leq x)$, as $n\rightarrow \infty$, for every real $x$ that is a continuity point of $P(X\leq x)$.
\end{itemize}  
\section{Main results}
\label{sec:main}

\subsection{Consistency and Central Limit Theorem for Deterministic Graph Streams}\label{subsection:det}
For deterministic temporal graph $T_m$, we establish conditions under which $\hat  C(H, T_m)/  C(H, T_m)$ converges to 1 in probability as $m \rightarrow \infty$, as shown in Theorem \ref{the:con-det}. To derive the asymptotic normality of the estimator we consider the rescaled statistic, and provide conditions under which the estimator is asymptotically normal in Theorem \ref{the:clt-det}. In this setting where $T_m$ is deterministic, edge sampling is the only source of randomness.

\begin{theorem}[Consistency]
Suppose $T_m = \{(u_i , v_i , t_i ) = (e_i, t_i)$, $i = 1, \cdots, m\}$ is a temporal graph on node set $V_m$. $H$ is a $k$-node, $l$-edge, $\delta$-temporal motif. If 
\begin{equation}\label{1}
   \frac{\sum_{i=1}^m\eta^2(e_i)} {(\sum_{i=1}^m\eta(e_i))^{2}} \rightarrow 0 \text{ as } m\rightarrow \infty,
\end{equation}
then 
\begin{displaymath}
\frac{\hat  C(H, T_m)}{C(H, T_m)} \xrightarrow{p}1.
\end{displaymath}
\label{the:con-det}
\end{theorem}

\begin{theorem}[Central Limit Theorem (CLT)]
Suppose $T_m = \{(u_i , v_i , t_i ) = (e_i, t_i)$, $i = 1, \cdots, m\}$ is a temporal graph on node set $V_m$. $H$ is a $k$-node, $l$-edge, $\delta$-temporal motif. If 
\begin{equation}\label{2}
\frac{\sum_{i=1}^m\eta^3(e_i)} {(\sum_{i=1}^m\eta^2(e_i))^{3/2}} \rightarrow 0 \text{ as } m\rightarrow \infty,
\end{equation}
then
\begin{displaymath}
Z(H, T_m):= \frac{\hat C(H, T_m) - C(H, T_m)}{\sqrt{var[\hat C(H, T_m)]}} \overset{d}{\rightarrow} N(0, 1 ), \text{ as } m \rightarrow \infty.
\end{displaymath}
\label{the:clt-det}
\end{theorem}

The proofs of Theorems \ref{the:con-det} and \ref{the:clt-det} are given in Appendix \ref{appendix:con-det} and \ref{appendix:clt-det}, respectively. To better understand the implications of the conditions in \eqref{1} and \eqref{2}, we offer two alternative assumptions, i.e., 
Assumption \ref{3} and Assumption \ref{4}, and show their relationship with the above consistency/CLT conditions, as summarized in Lemma \ref{lem-1}.

\begin{assumption} \label{3}
For sequence $\{\eta(e_i)\}_{i=1}^m$,  $\eta(e_i) = O(1)$ (i.e., $\eta(e_i) \leq C_2 < \infty$ for $i$ large enough, where $C_2$ is a constant), and $C(H, T_m) \gg m^{1/2}$ (i.e., $\lim_{m\rightarrow \infty} m^{1/2}/C(H, T_m) = 0$), as $m\rightarrow \infty.$
\end{assumption}

\begin{assumption} \label{4}
For sequence $\{\eta(e_i)\}_{i=1}^m$,  $\eta(e_i) = O(1)$, and $C(H, T_m) \gg m^{2/3}$, as $m\rightarrow \infty.$
\end{assumption}

\begin{lemma}
Assumption \ref{3} implies assumption \eqref{1}. Assumption \ref{4} implies assumption \eqref{1} and \eqref{2}. 
\label{lem-1}
\end{lemma}

The proof of Lemma \ref{lem-1} is provided in Appendix \ref{appendix_proof_lemma}. 
Combining Theorem \ref{the:con-det}, Theorem \ref{the:clt-det} and Lemma \ref{lem-1}, we can see that to assure consistency and a CLT for the count estimator in \eqref{estimator} 
requires sufficiently frequent appearance of the given temporal motif $H$ in $T_m$, but without any particular edge(s) dominating. Note that consistency requires $C(H, T_m) \gg m^{1/2}$, while the CLT requires $C(H, T_m) \gg m^{2/3}$ -- basically higher frequency of motif appearance is needed for the CLT than for consistency.

\textit{Remark.} The consistency and CLT for $\hat C(H, T_m)$ provided in Theorem \ref{the:con-det} and Theorem \ref{the:clt-det} are established under the assumption that the sampling probability $p$ is a constant. It might be of interest to have it vary with $m$. Consistency and CLT results can also be obtained under such a scenario. The sufficient conditions will change to $C(H, T_m) \gg p_m^{-1/2}m^{1/2}$ for consistency to hold and $C(H, T_m) \gg [p_m(1-p_m)]^{-1/3}m^{2/3}$ for CLT to hold. These conditions can be naturally derived from the proof for constant $p$. 

We now use the results above to construct asymptotically valid confidence intervals for the count $C(H, T_m)$, as shown in \eqref{equ:CI} below. To this end, a consistent estimator for the variance of the count estimator is introduced (\eqref{equ:var} below). These results are summarized in Proposition \ref{prop:ci}, the proof of which is provided in Appendix \ref{appendix_proof_prop}
\begin{prop}
For a temporal graph  $T_m = \{(u_i , v_i , t_i ) = (e_i, t_i)$, $i = 1, \cdots, m\}$, a $k$-node, $l$-edge $\delta$-temporal motif $H$, and a sampling ratio $p$. Suppose $\eta(e_i) = O(1)$ and $C(H, T_m) \gg m^{2/3}$. Then as $m \rightarrow \infty$, the following holds. 
\begin{enumerate}
    \item Let 
    \begin{equation}\label{equ:var}
        \hat \sigma^2(H, T_m) := \frac{1-p}{p^2_m l^2}  \sum_{i = 1}^m \omega_i \eta^2(e_i).
    \end{equation}
    Then $\hat \sigma^2(H, T_m)$ is a consistent estimate of $\sigma^2(H, T_m): =  \frac{1-p}{p l^2}  \sum_{i = 1}^m \eta^2(e_i)$, that is, $\frac{\hat \sigma^2(H, T_m)}{ \sigma^2(H, T_m)} \rightarrow_P 1 $.
    \item 
    \begin{equation}\label{equ:CI}
    \begin{split}
        P\Big( C(H, T_m) \in \big [ 
        \hat C(H, T_m) - z_{\alpha/2} \hat \sigma(H, T_m), & \\ \hat C(H, T_m) + z_{\alpha/2} & \hat \sigma(H, T_m)
        \big]\Big) 
        \rightarrow 1-\alpha \enskip ,
    \end{split}
    \end{equation}
    where $z_{\alpha/2}$ is the $(1-\frac{\alpha}{2})$-th quantile of the standard normal distribution $N(0, 1)$.  
\end{enumerate}
\label{prop:ci}
\end{prop}

\subsection{Consistency and CLT for a Stochastic Graph Stream}\label{sec:sto} 
The results above show that certain characteristics of local/global motif counts for a deterministic temporal network are sufficient to ensure that consistency and a CLT hold for the count estimator under edge sampling. If the temporal networks are instead realizations from a stochastic model, it is natural to ask under what assumptions with respect to that model the consistency and asymptotic normality continue to hold, where the underlying randomness is now in both the edge sampling and the generation of the network. 

In this section, we leverage the results in Section \ref{subsection:det} to derive sufficient
conditions for the desired asymptotic behavior under a stochastic model for temporal graphs that is essentially a classical random graph with Poisson arrivals of edges.
Instead of $T_m$, we use $T(\tau) = \{(e_i, t_i), i=1,\cdots, N(\tau)\}$ to represent a random temporal graph, with a random number of edges $N(\tau)$ occurring in the time interval $(0,\tau]$, in which we assume 
\begin{enumerate}[label=(\roman*)]
		\item
 		$t_1, t_2,\cdots, t_{N(\tau)}$ are arrival times from a Poisson process with rate $\lambda$, and
		\item 
		each $e_i$ is sampled uniformly among the set of all node pairs, at each arrival time $t_i$.
\end{enumerate}
Hence, $N(\tau) \sim Poisson(\lambda \tau)$. Here we state the main result as Theorem \ref{the:clt-st}, the proof of which is deferred to Appendix \ref{appendix:proof_clt-st}.

\begin{theorem}
[Asymptotics under marked Poisson point process]
\label{the:clt-st}
Define an observation time interval $[0, \tau]$. Let $T(\tau) = \{ (e_i, t_i)$, $i = 1, \cdots, N(\tau)\}$ be a random temporal graph as defined above. Let $|V|$ be the total number of vertices in $T(\tau)$. Then for any $k$-node, $l$-edge $\delta$-temporal motif $H$, if 
the following hold,
\begin{enumerate}
    \item $\lambda>0$, and $\lambda = O(1)$, as $\tau \rightarrow \infty $,
    \item $|V|<\infty$ is fixed,
 as $\tau \rightarrow \infty $,
\end{enumerate}
it follows that 
$\hat C(H, T(\tau))/C(H, T(\tau)) \xrightarrow{p}1$ 
and the rescaled statistic $Z(H, T(\tau)) \overset{d}{\rightarrow}  N(0,1), \text{ as }\tau\rightarrow \infty$.
\end{theorem}

Theorem \ref{the:clt-st} shows that the motif count estimator is consistent and asymptotically normal under this Poisson counting process model whenever the rate of the process is positive and bounded, and the total number of vertices in the edge stream is fixed as the stream continues for an infinite amount of time. There is a rich literature on modeling edge streams by counting process (see \cite{matias2018semiparametric} for a summary). These assumptions on the rate of edge appearance in the graph stream are more practical for validation on real data than the previous assumptions on the magnitude of true motif counts. Note that the above CLT established for the Poisson counting process model also ensures the validity of the confidence intervals for the motif count shown in \eqref{equ:CI} under this graph stream model.

\section{Numerical Illustration}
\label{sec:numerical}

In this section, we perform experiments on both synthetic and real-world data. Using synthetic data, we illustrate the practical impact of consistency and asymptotic normality of the count estimator under both deterministic and stochastic cases. We also evaluate the coverage probability of the confidence intervals on a real data set with different sampling ratios. Code for reproducing our simulation is available at \url{https://github.com/KolaczykResearch/TempMotifEstim}.

\subsection{Simulation} 

\subsubsection*{Deterministic Case}
We consider the sampling distribution of the motif count estimator for a fixed temporal network, which is chosen to be a realization from a random graph model. We present results for two random temporal graph models. 

We first choose this fixed temporal network to be a realization from a homogeneous Poisson process model with uniformly chosen edges as described in Section \ref{sec:sto}. We simulate one temporal network with length $m = 7000$ from the model on $|V| = 100$ vertices for each rate $\lambda$ varying between $25$ and $250$. For each simulated network, we sample with probability $0.03$ and estimate the number of cyclic triangles (shown as (d) in Figure \ref{query_motifs}) in it repeatedly and calculate mean and standard deviation of $\hat C(H, T_m)/C(H, T_m)$ over $100$ replications. Results are shown in Figure \ref{fig:1_det_m1} (Left) for a range of 8 rate values. We can see that as $\lambda$ increases, the mean stabilizes at $1$ and the size of the error bar decreases, signaling a better estimate of $C(H, T_m)$. This is because more edges are expected to occur under higher $\lambda$ value, thus more temporal triangles are expected to be formed within a certain time duration, and one of the conditions that determines the consistency of $\hat C(H, T_m)$ is a relatively high frequency of appearance for the target motif. Figure \ref{fig:1_det_m1}  (Right) illustrates the asymptotic normality of the motif count estimator for the number of cyclic triangles in a fixed temporal network with length $m=200000$, $|V|=100$ vertices simulated under rate $\lambda=250$. We set the sampling ratio to be $0.03$, and plot the histogram of $\hat C(H, T_m)/C(H, T_m)$ over $5000$ replications, which, as expected, is centered around $1$ and aligns well with the normal density curve shown in blue.

We then choose this fixed temporal network to be a realization from a Poisson process stochastic block model \cite{matias2018semiparametric} where interactions between each pair of nodes are counted by a homogeneous Poisson process with intensity driven by the community structure of nodes. This model is a natural extension of the standard stochastic block model to the case of multivariate counting processes for recurrent interaction events. We simulate one temporal network with length $m=7000$ from the model on $|V|=100$ vertices with two equally sized blocks, fixed off-diagonal intensity of $0.06$, for each diagonal intensity varying between $0$ and $0.2$, where diagonal intensity refers to the intensity for interactions between two nodes from the same group, and off-diagonal intensity refers to that from different groups. For each simulated network, we sample with probability $0.03$ and estimate the number of cyclic triangles in it repeatedly and calculate mean and standard deviation of $\hat C(H, T_m)/C(H, T_m)$ over $100$ replications. Results are shown in Figure \ref{fig:2_det_m2} (Left) for a range of $8$ diagonal intensity values. We can see that as the diagonal intensity increases, the mean stabilizes at $1$ and the size of the error bar decreases, signaling a better estimate of $C(H, T_m)$. This is also because more temporal triangles are expected to be formed within a certain time duration under higher diagonal intensity where individuals from the same group interact more frequently, and again one of the conditions that determines the consistency of $\hat C(H, T_m)$ is a relatively high frequency of appearance for the target motif. Figure \ref{fig:2_det_m2} (Right) illustrates the asymptotic normality of the motif count estimator under this setting.

\begin{figure}[htbp]
\centering 
\includegraphics[scale=0.4]{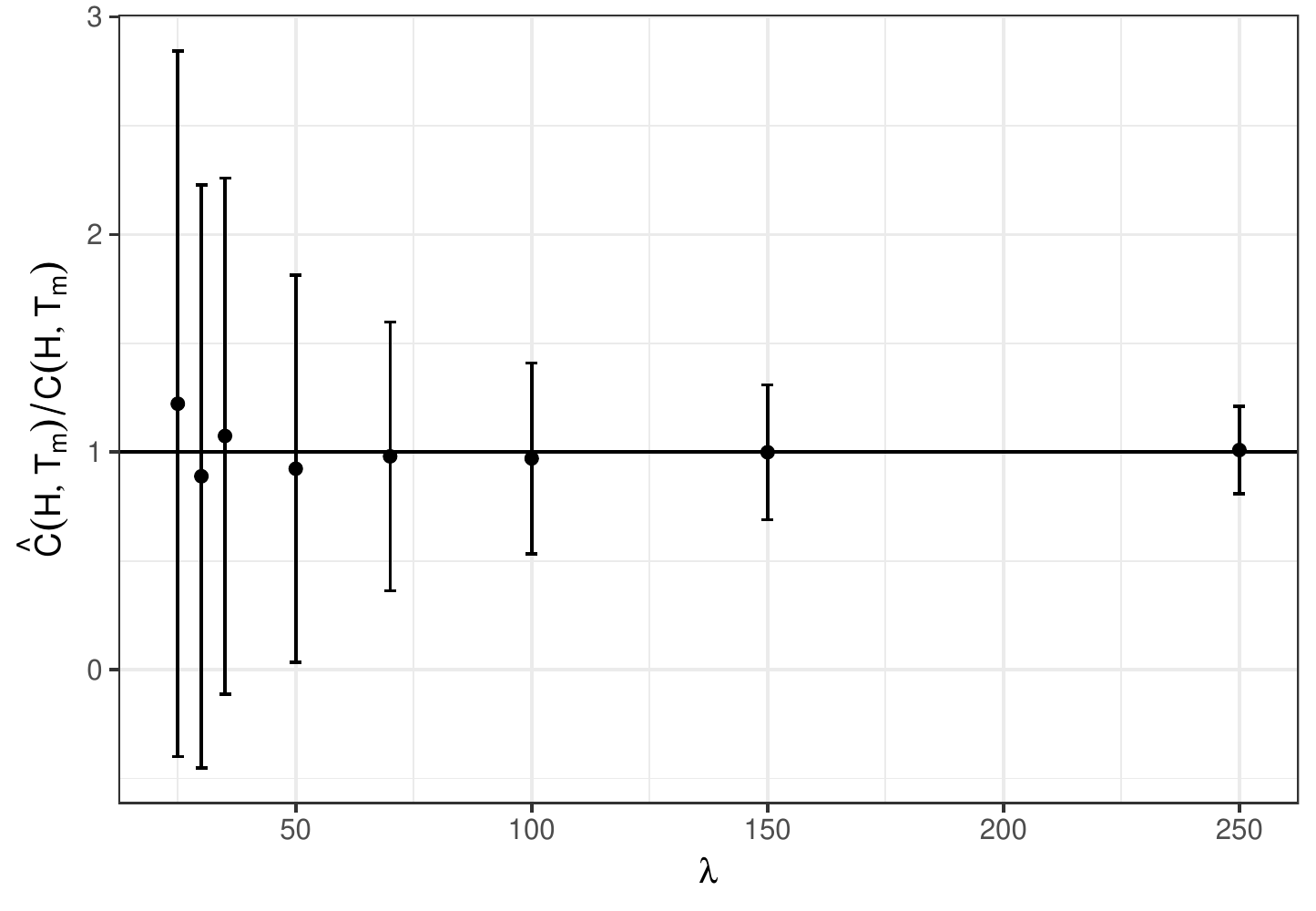}
\includegraphics[scale=0.4]{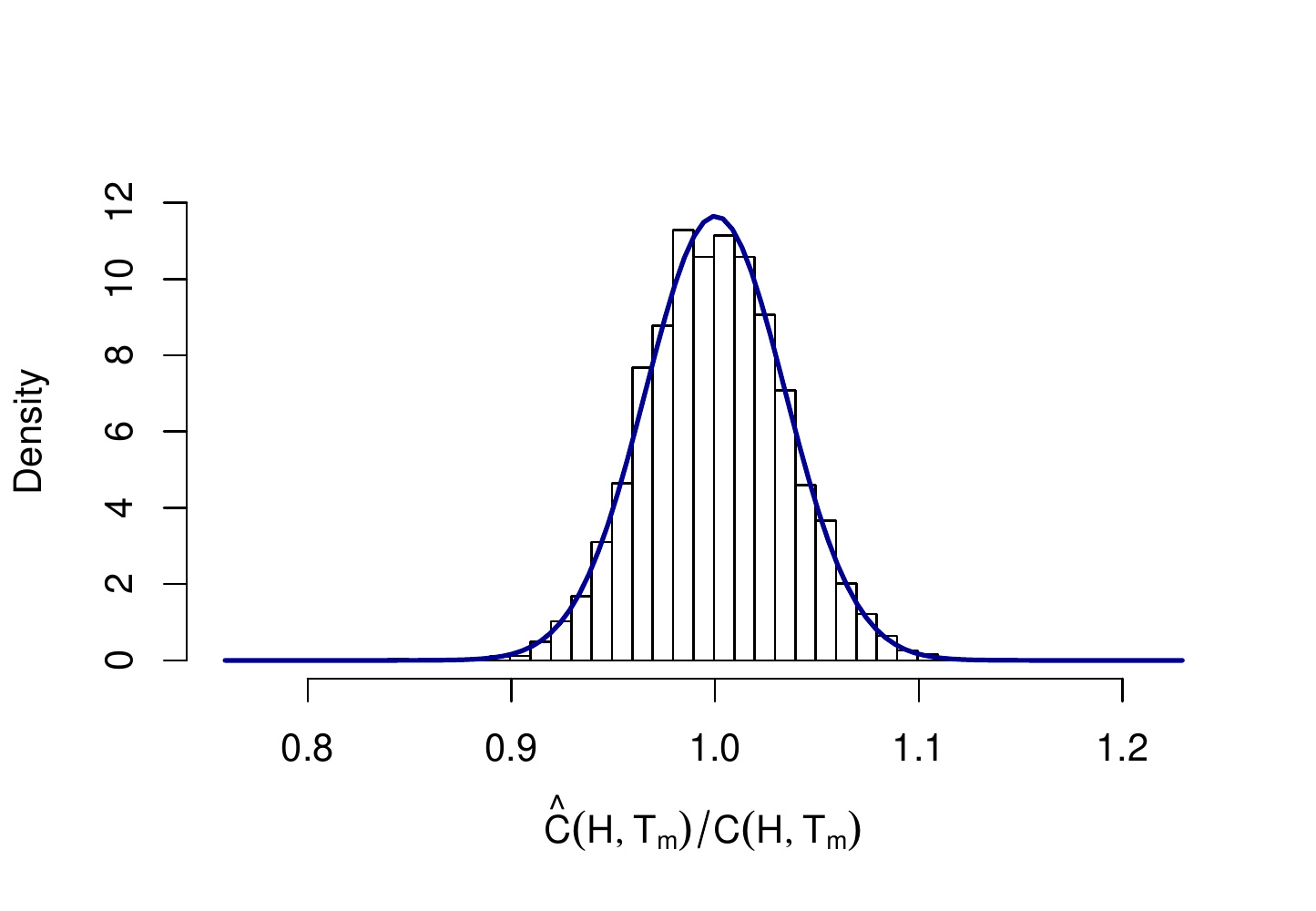}
\caption{Results for deterministic networks simulated from Poisson process uniform model. {\normalfont Left}: Empirical 1-standard deviation error bars for $\hat C(H, T_m)/ C(H, T_m)$ in estimating the number of temporal triangles (directed and cyclic, $\delta=2$) for each deterministic network generated from a Poisson process uniform model with rate $\lambda$ ranging from $25$ to $250$. Sampling ratio $p=0.03$, $100$ replications, $|V|=100, m=7000$. {\normalfont Right}: Histogram of $\hat C(H, T_m)/ C(H, T_m)$ in Poisson process uniform model with rate $\lambda=250$, sampling ratio $p=0.03$, $5000$ replications, $|V|=100, m=200000$, and the limiting normal density in blue.}
\label{fig:1_det_m1}
\end{figure}

\begin{figure}[htbp]
\centering 
\includegraphics[scale=0.4]{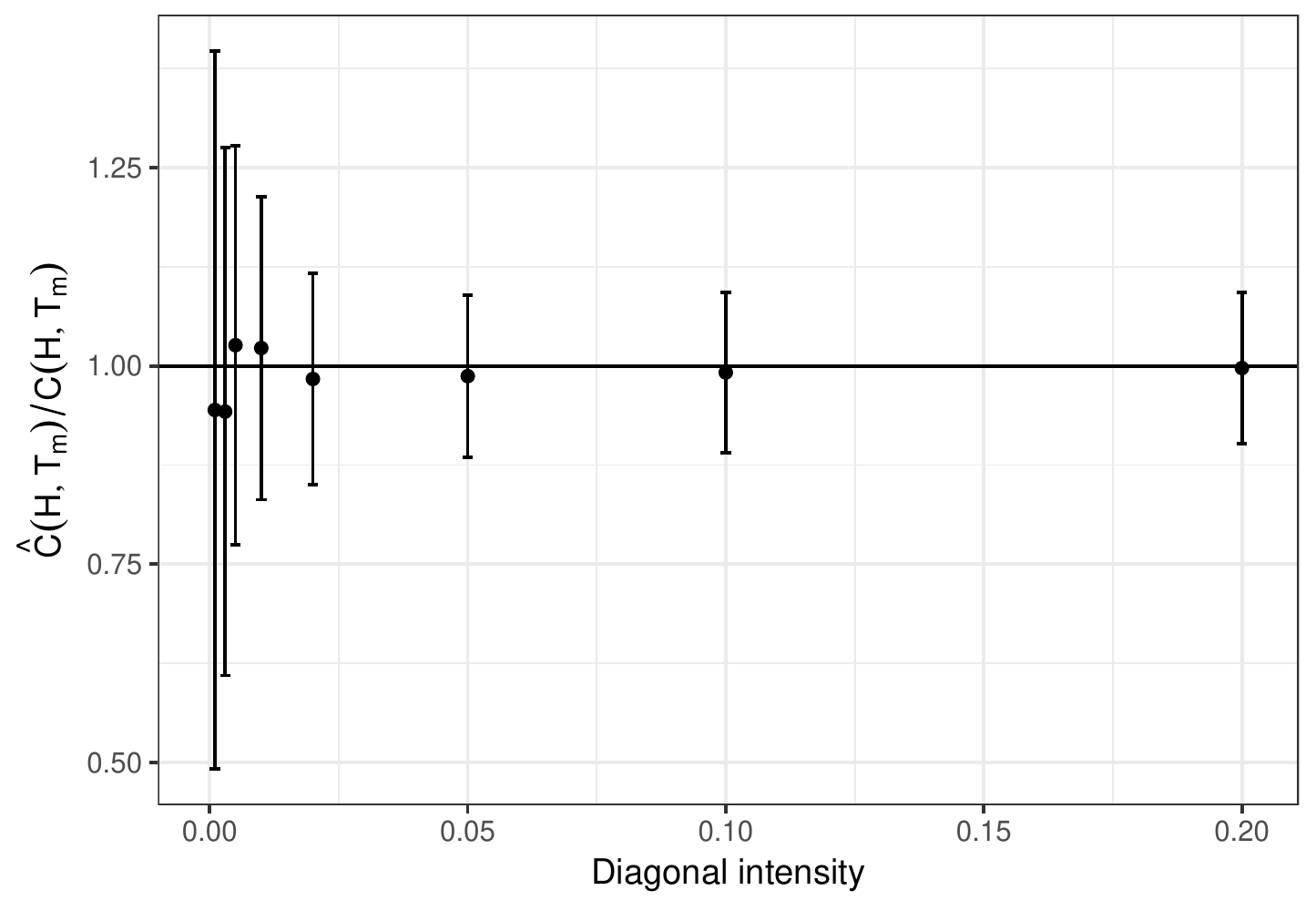}
\includegraphics[scale=0.4]{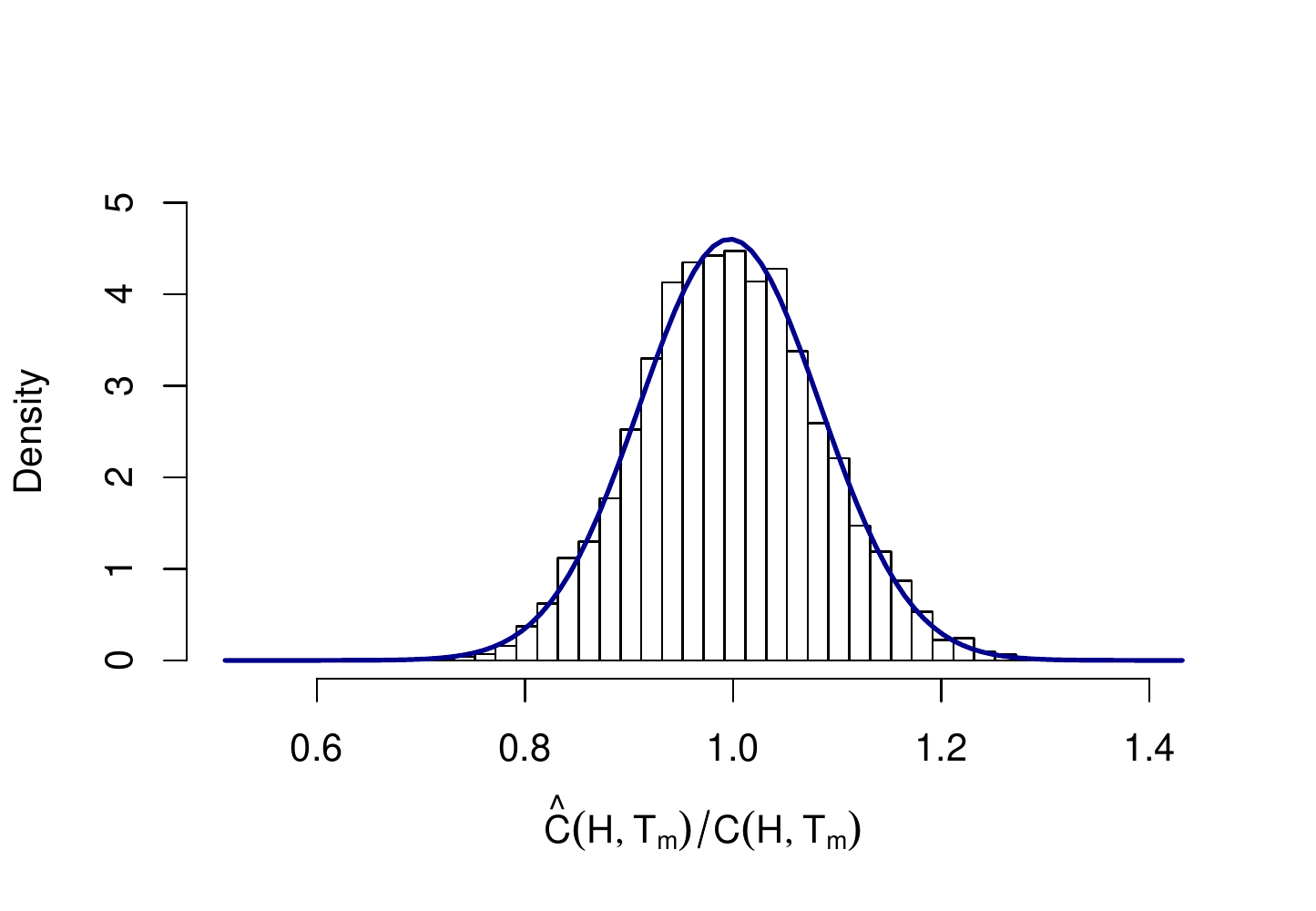}
\caption{Results for deterministic networks simulated from Poisson process stochastic block model. {\normalfont Left}: Empirical 1-standard deviation error bars for $\hat C(H, T_m)/ C(H, T_m)$ in estimating the number of temporal triangles (directed and cyclic, $\delta=2$) for each deterministic network generated from a homogeneous Poisson process stochastic block model with 2 blocks, equal block size, intensity driven by the individuals’ groups memberships, off-diagonal intensity of $0.06$,  and diagonal intensity varying between $0$ and $0.2$. Sampling ratio $p=0.03$, $100$ replications, $|V|=100, m=7000$. {\normalfont Right}: Histogram of $\hat C(H, T_m)/ C(H, T_m)$ in homogeneous Poisson process stochastic block model with 2 blocks, equal block size, off-diagonal intensity $0.06$, diagonal intensity $0.2$, sampling ratio $p=0.03$, $5000$ replications, $|V|=100, m=7000$, and the limiting normal density in blue.}
\label{fig:2_det_m2}
\end{figure}

\subsubsection*{Stochastic Case} 
We also perform experiments under the two random temporal graph models for stochastic case. 

We first simulate temporal networks from the homogeneous Poisson process model with uniformly chosen edges as described in Theorem \ref{the:clt-st}, and assess properties of the count estimator as observation time $\tau$ increases. Note that in the stochastic case, we generate a new stream of edges from the temporal graph model with rate $\lambda=30$ in observation time interval $[0, \tau]$ for each sampling and motif estimation trial, compared to the deterministic case where the stream of edges is fixed over all trial replications. Figure \ref{fig:3_sto_m1} (Left) shows the empirical means and standard deviations of $\hat C(H, T(\tau))/C(H, T(\tau))$ in estimating the number of cyclic triangles calculated over $100$ replications (i.e., $100$ realizations of temporal networks) as a function of observation time $\tau$ ranging from $200$ to $5000$. We see that as the observation time increases, the mean gets stabilized at $1$ and the size of the error bar decreases, illustrating the consistency of the count estimator as $\tau \rightarrow 0$. Meanwhile, Figure \ref{fig:3_sto_m1} (Right) shows a histogram of $\hat C(H, T_m)/C(H, T_m)$ over $5000$ replications for observation time $\tau=9000$, which is centered around $1$ and aligns well with the normal density curve shown in blue, illustrating the asymptotic normality of the motif count estimator under this stochastic model.

We also simulate temporal networks from the Poisson process stochastic block model, and assess properties of the count estimator as observation time $\tau$ increases. Multiple sequences of edges are simulated from the temporal graph model, one for each sampling and motif estimation trial. Figure \ref{fig:4_sto_m2} shows the results from the repeated trials, where the left plot illustrates the consistency of the count estimator as $\tau \rightarrow 0$, and the right plot shows the asymptotic normality of the motif count estimator under a particular parameter setting of this stochastic model.

 \begin{figure}[htbp]
\centering 
\includegraphics[scale=0.4]{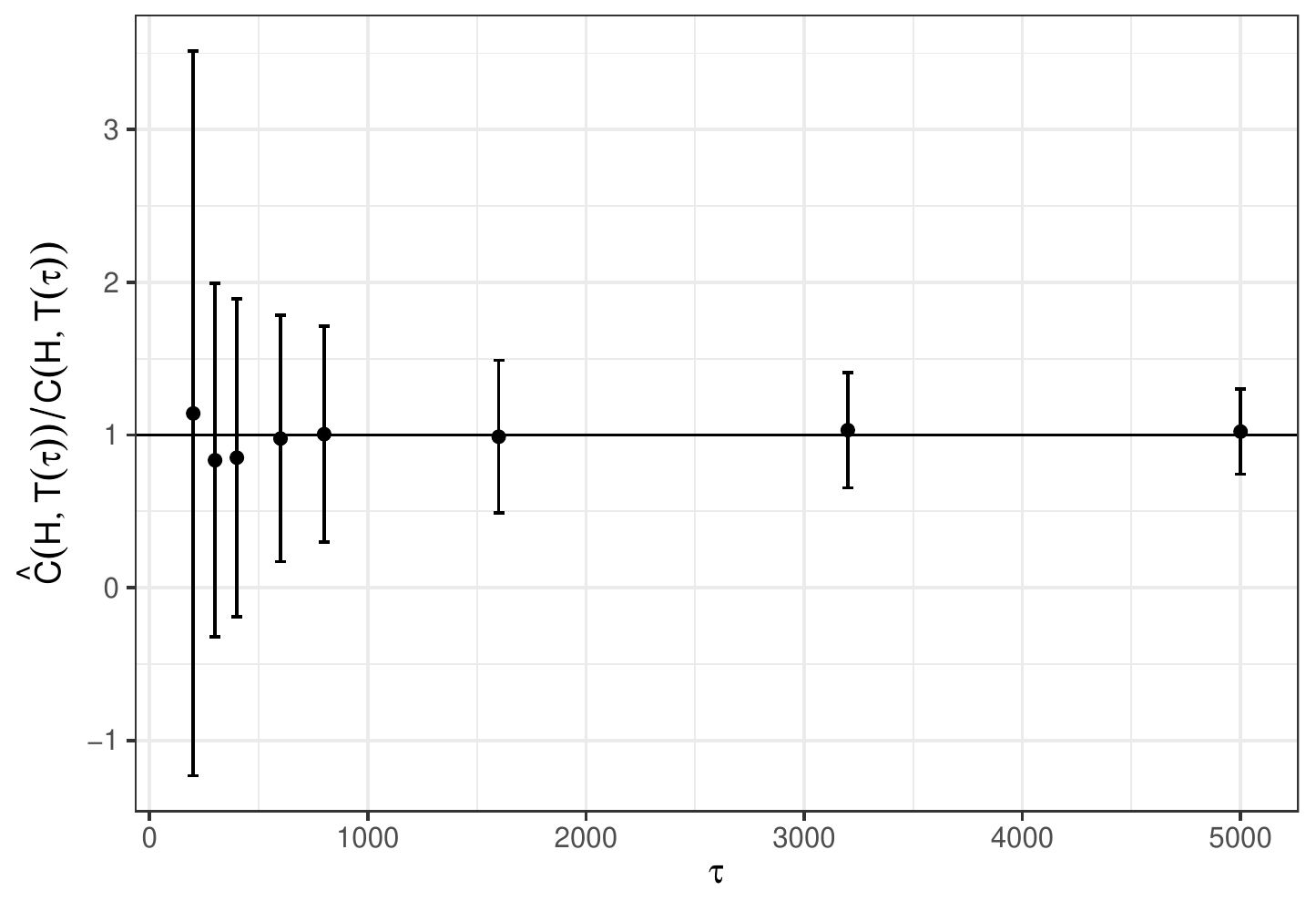}
\includegraphics[scale=0.4]{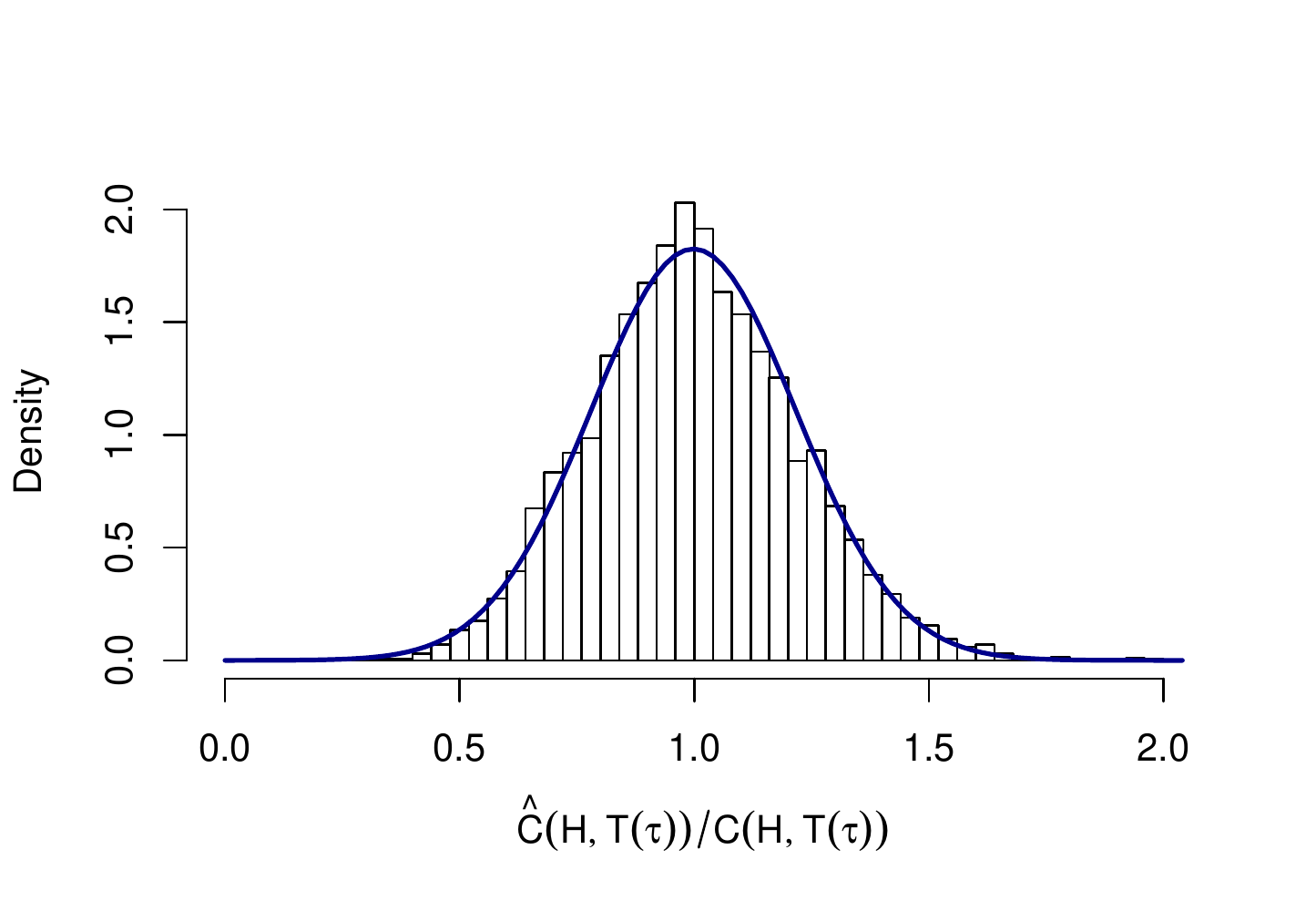}
\caption{Results for stochastic networks simulated from Poisson process uniform model. {\normalfont Left}: Empirical 1-standard deviation error bars for $\hat C(H, T(\tau))/ C(H, T(\tau))$ in estimating the number of temporal triangles (directed and cyclic, $\delta=2$) in a Poisson process uniform model with rate $\lambda=30$ in observation time interval $[0,\tau]$. $\tau$ is varying between $200$ and $5000$. sampling ratio $p=0.03$, $100$ replications, $|V|=100$. {\normalfont Right}: Histogram of $\hat C(H, T(\tau))/ C(H, T(\tau))$ in Poisson process uniform model with rate $\lambda= 30$, $|V|=100, \tau=9000$, sampling ratio $p=0.03$, $5000$ replications, and the limiting normal density in blue. }
\label{fig:3_sto_m1}
\end{figure}

 \begin{figure}[htbp]
\centering 
\includegraphics[scale=0.4]{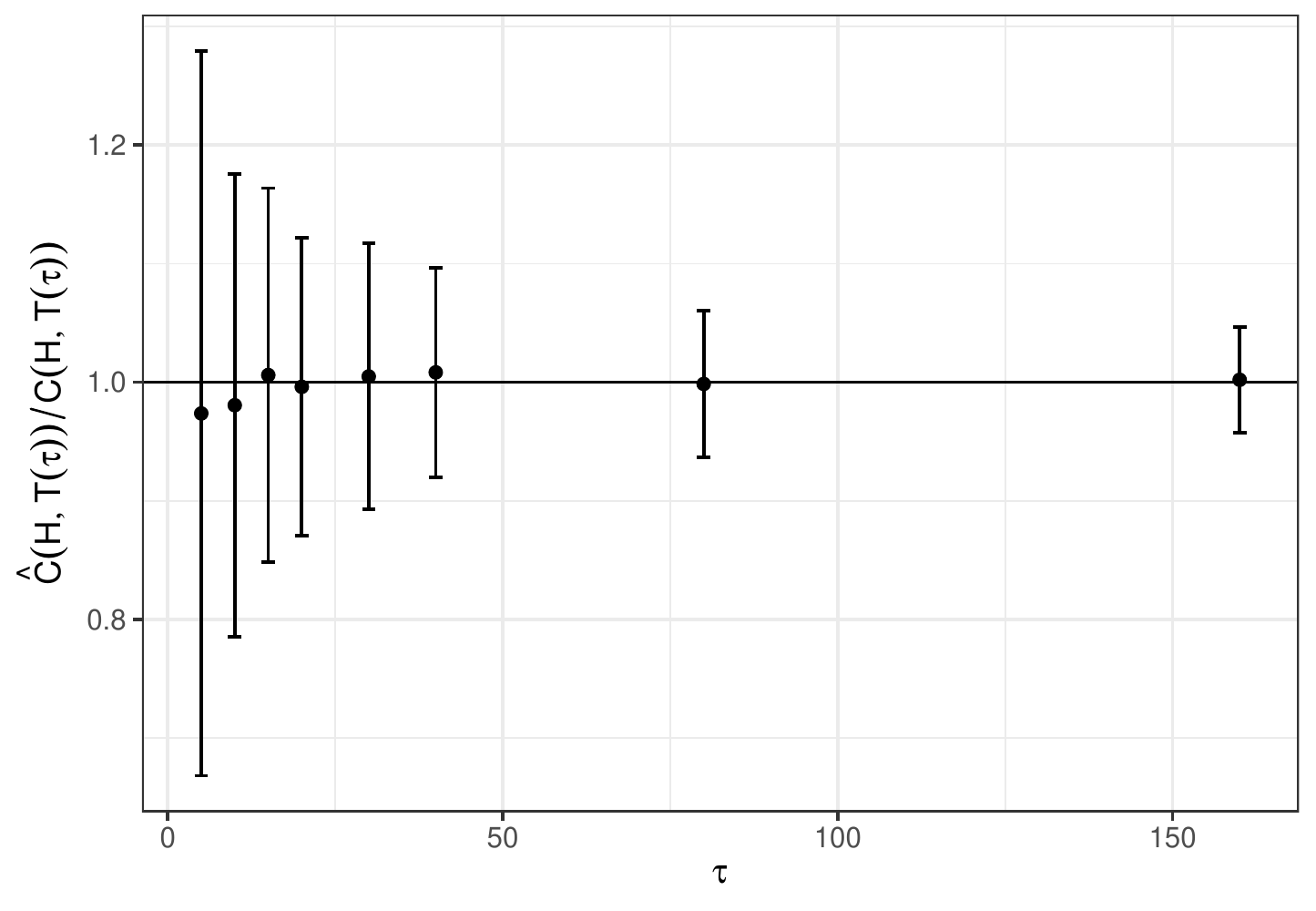}
\includegraphics[scale=0.4]{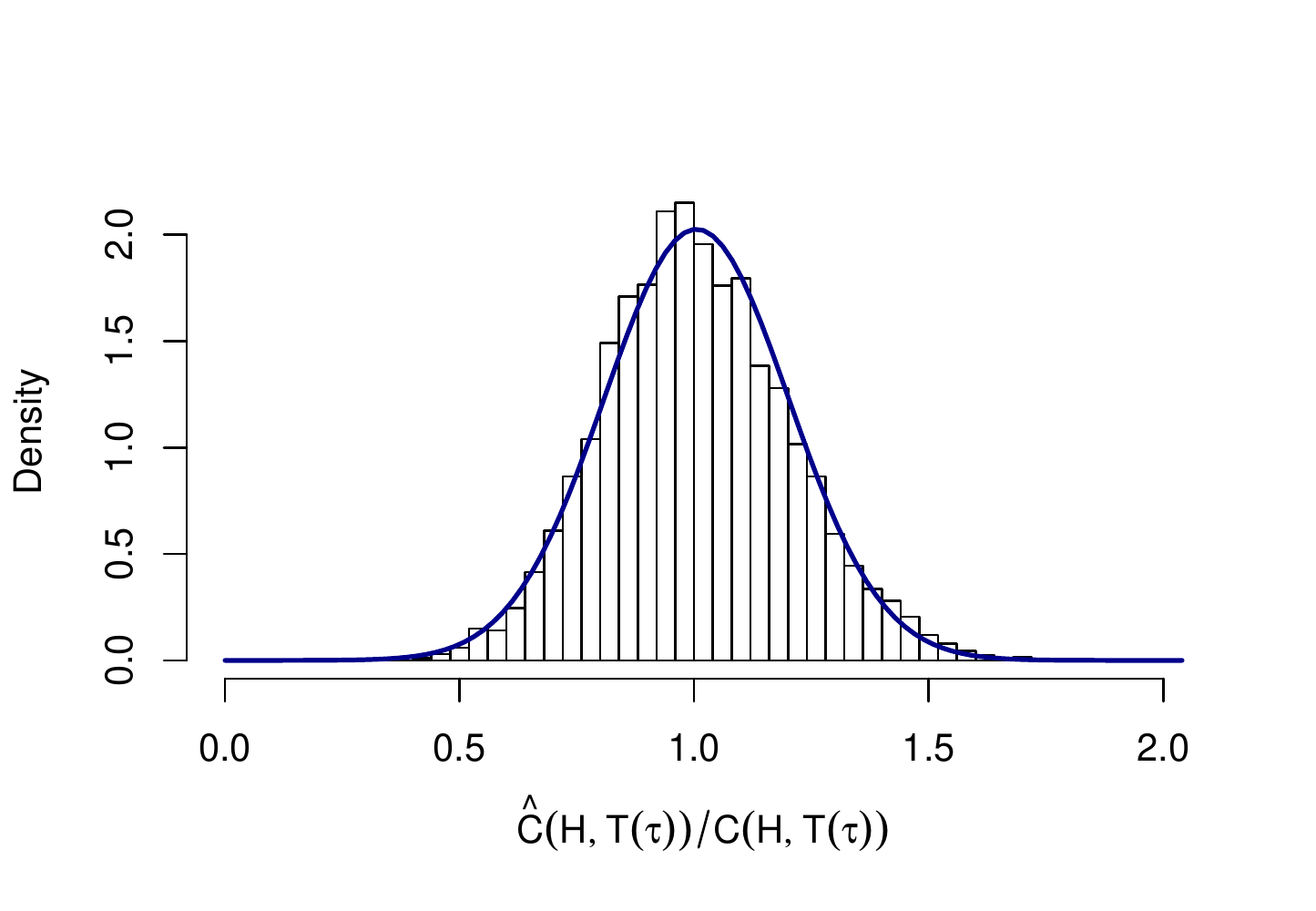}
\caption{Results for stochastic networks simulated from Poisson process stochastic block model. {\normalfont Left}: Empirical 1-standard deviation error bars for $\hat C(H, T_m)/ C(H, T_m)$ in estimating the number of temporal triangles (directed and cyclic, $\delta=2$) in a homogeneous Poisson process stochastic block model with 2 blocks, equal block size, intensity driven by the individuals’ groups memberships, off-diagonal intensity $0.06$, and diagonal intensity $0.02$ in observation time interval $[0,\tau]$. $\tau$ is varying between $5$ and $160$. Sampling ratio $p=0.03$, $100$ replications, $|V|=100$. {\normalfont Right}: Histogram of $\hat C(H, T(\tau))/ C(H, T(\tau))$ in Poisson process stochastic block model with off-diagonal intensity $0.06$, diagonal intensity $0.02$, $|V|=100, \tau=10$, sampling ratio $p=0.03$, $5000$ replications, and the limiting normal density in blue.}
\label{fig:4_sto_m2}
\end{figure}

\subsection{Application}
In this subsection, we perform experiments on real-world data to evaluate the coverage probability of the estimated $95\%$ confidence interval for the motif count under different sampling ratios for six types of motifs, as shown in Figure \ref{query_motifs}. In this experiments, we use messaging temporal network data \cite{panzarasa2009patterns}, which is comprised of private messages sent on an online social network at the University of California, Irvine. This data includes $59835$ interactions among $1899$ college students within a time span of 193 days. We set the time span $\delta$ for the motif to $86400$ seconds $= 1$ day. For each query motif, under one of the sampling ratios in $\{0.01, 0.03, 0.05, 0.1, 0.2\}$, we estimate the coverage probability of $95\%$ confidence intervals for motif count $C(H, T_m)$ using the relative frequency (RF) of the event that the confidence interval covers the true motif count in $5000$ sampling replications. Figure \ref{fig:app} shows the estimated coverage probability as a function of sampling ratio $p$ for the six query motifs. We see that the empirical coverage probabilities of CIs fall below the declared confidence level under low sampling ratios, and get closer as the sampling ratio increases. Also notice that the empirical coverage of the CIs estimated for motif $H_1$, $H_2$ and $H_3$ (shown as (a), (b), (c) in Figure \ref{query_motifs}) are better than that for motif $H_4$, $H_5$ and $H_6$ (shown as (d), (e), (f) in Figure \ref{query_motifs}). This is because the true counts for the former three motifs (which are $381720, 1201092, 295970$, respectively) are greater than that for the latter three (which are $9850, 16064, 271022$., respectively). This is a reflection of the role of the CLT condition for $\hat C(H, T_m)$ that requires frequent appearance for the target motif.

\begin{figure}[htbp]
\centering 
\includegraphics[scale=0.3]{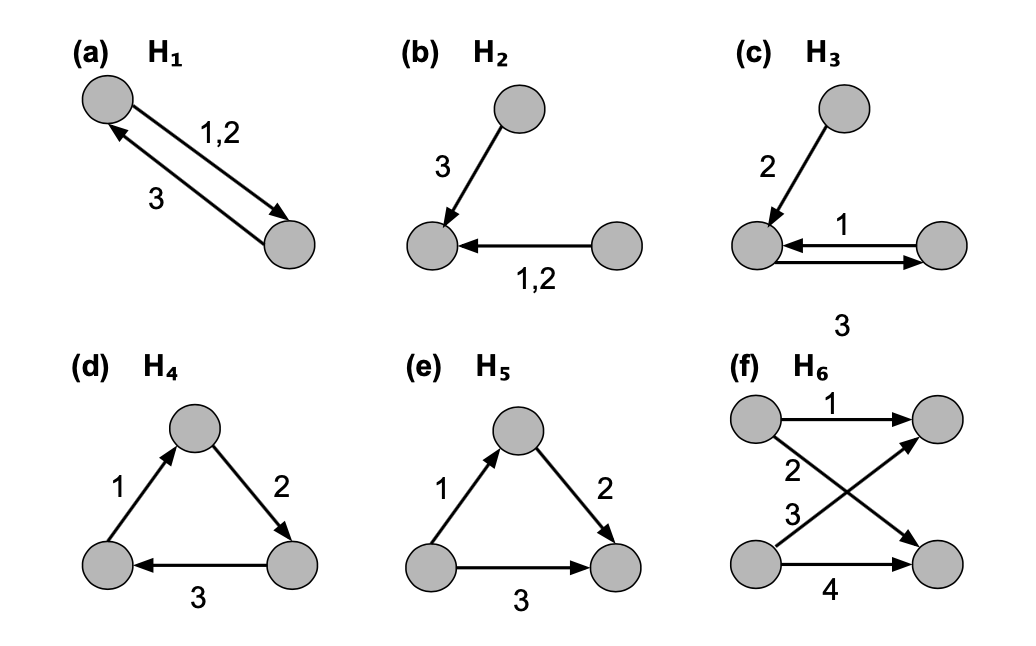}
\caption{Six query motifs. The numbers on edges represent edge ordering.}
\label{query_motifs}
\end{figure}

\begin{figure}[htbp]
\centering 
\includegraphics[scale=0.45]{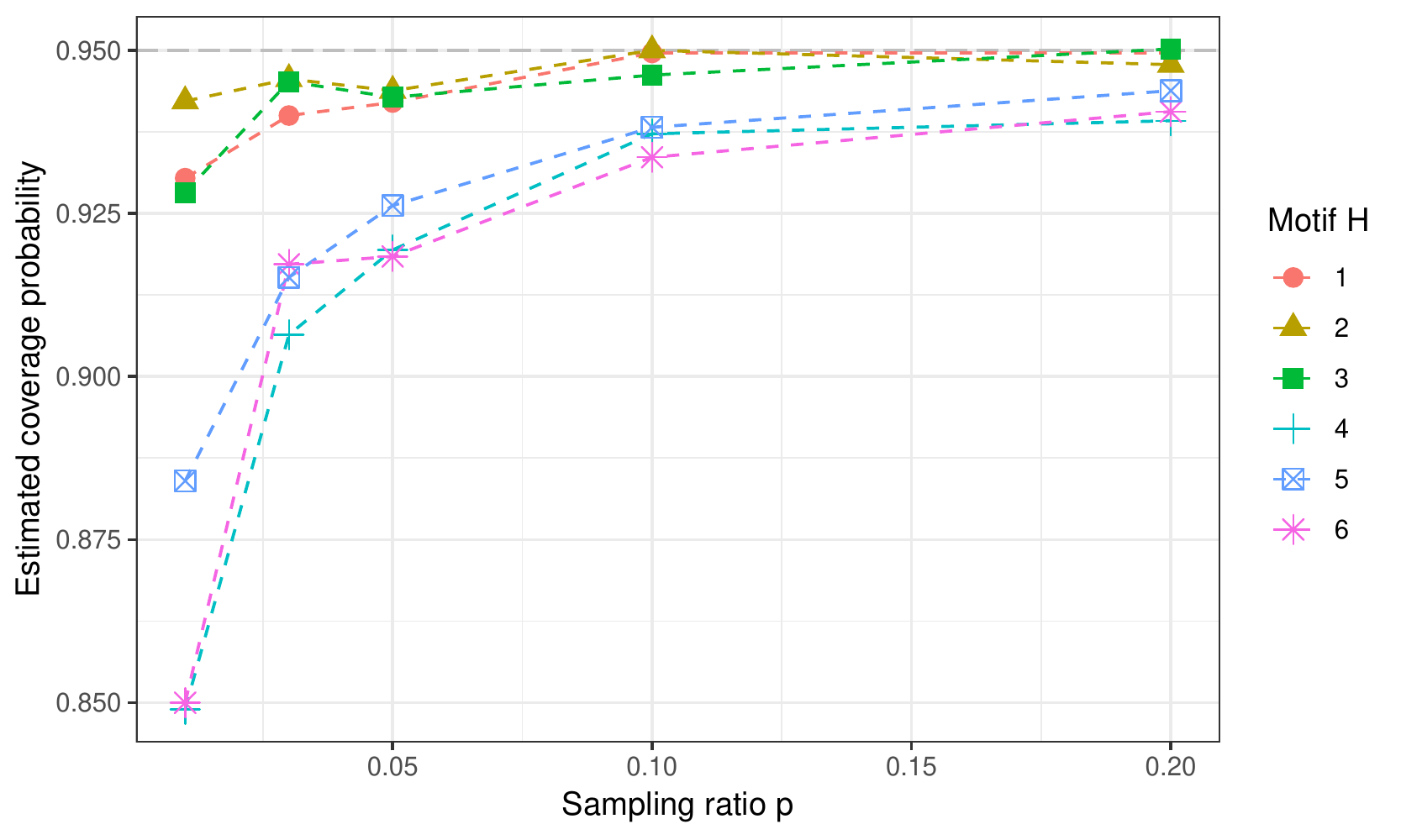}
\caption{Estimated coverage probability of $95\%$ confidence intervals for motif count $C(H, T_m)$ from $5000$ sampling replications for college messaging temporal network. Reported in the plots are the relative frequencies (RF) of the event that a confidence interval covers the
corresponding true motif count for one of the motifs in Figure \ref{query_motifs} under sampling ratio $p\in \{0.01, 0.03, 0.05, 0.1, 0.2\}$.}
\label{fig:app}
\end{figure}

\section{Discussion}
\label{sec:discussion}

In this work, we consider the problem of temporal motif estimation in graph streams under an edge sampling model. We establish conditions under which the count estimator $\hat  C(H, T_m)$ is consistent and asymptotically normal as $m \rightarrow \infty$, for deterministic $T_m$, and construct asymptotically valid confidence intervals for the motif count $C(H, T_m)$. We also theoretically derive the conditions w.r.t. a simple counting process model of temporal networks for consistency and asymptotic normality of the count estimator. Simulation studies are conducted to illustrate numerically the consistency and asymptotic normality of the motif estimator under both deterministic and stochastic cases. We also evaluate the true coverage probability of the confidence interval using real data for a variety of motifs.

There are two interesting directions for future work. One is to study the asymptotics of the temporal motif estimator under other sampling models, e.g., the subwindow sampling regime in \cite{liu2018sampling} and \cite{sarpe2021presto}, where smaller time intervals are sampled from the graph streams, and exact motif counts in the sampled time intervals are used to estimate the total counts. The other direction is to study the asymptotics under more complicated settings of graph stream models, e.g., inhomogeneous counting processes, and correlated multivariate point process for non-independent occurrences of interaction events \cite{perry2013point}. These adaptions are relevant to the characteristic of ``burstiness'' often seen in empirical dynamic network data sets \cite{min2013burstiness}.

\appendix
\section{Proof of Theorem \ref{the:con-det}} \label{appendix:con-det}
By Chebyshev's inequality, for every $\epsilon > 0$, we have 
\begin{equation}
\begin{split}
P\left(\big | \frac{\hat  C(H, T_m)}{C(H, T_m)} - 1\big| > \epsilon \right) & = P\left (|\hat  C(H, T_m)- C(H, T_m)|  > \epsilon \cdot C(H, T_m) \right) \\
& \leq \frac{var(\hat  C(H, T_m))}{\epsilon^2  C^2(H, T_m)} =  \frac{1-p}{\epsilon^2 p l^2}  \sum_{i = 1}^m \eta^2(e_i)\cdot \frac{l^2}{(\sum_{i = 1}^m \eta(e_i))^2} \\
& =  \frac{1-p}{\epsilon^2 p } \frac{\sum_{i = 1}^m \eta^2(e_i)}{(\sum_{i = 1}^m \eta(e_i))^2}  \rightarrow 0, \text{ as }m \rightarrow \infty.
\end{split}
\end{equation}

\section{Proof of Theorem \ref{the:clt-det}} \label{appendix:clt-det}

To derive the asymptotic normality of the estimator we consider the rescaled statistic 
\begin{equation}
Z(H, T_m):= \frac{\hat C(H, T_m) - C(H, T_m)}{\sqrt{var[\hat C(H, T_m)]}} .
\end{equation}

\begin{proof}
Assume $\sigma^2(H, T_m) := var(\hat C(H, T_m)) =  \frac{1-p}{p l^2}  \sum_{i = 1}^m \eta^2(e_i)>0$. To simplify the notation, we drop the dependency on $H$ and $T_m$ from $\sigma(H, T_m)$, $Z(H, T_m)$ and denote them by $\sigma$ and $Z$, respectively when no confusion is possible. Define 
\begin{equation}
Y_i := \frac{\eta(e_i)}{pl}(\omega_i -p) 
\end{equation}
for $i=1,2,\cdots,m$. 
Then 
\begin{equation}
\begin{split}
Z &= \frac{1}{\sigma}(\hat C(H, T_m) - C(H, T_m))  = \frac{1}{\sigma} 
\sum_{i=1}^m (
\frac{1}{pl} \omega_i\eta(e_i) -\frac{1}{l} \eta(e_i))  \\
& = \frac{1}{\sigma} \sum_{i=1}^m Y_i.
\end{split}
\end{equation}
Note that $Y_i$'s are independent of each other with $\mu_i := E[Y_i] = 0$, $\sigma_i^2 := var[Y_i] =\frac{1-p}{pl^2} \eta^2(e_i)$. Let $S_m := \sum_{i=1}^m Y_i$, $s_m^2 := \sum_{i=1}^m \sigma_i^2 =  \frac{1-p}{pl^2} \sum_{i=1}^m \eta^2(e_i) = \sigma^2 $, then $Z = S_m/s_m $. Using Berry-Esseen's Theorem (Theorem 6.2 in Chapter 7 of \cite{gut2005probability}), we have  
\begin{equation}
\begin{split}
    \sup_{x\in \mathbb{R}} |F_Z(x) - \Phi(x)| \leq C \cdot \frac{\sum_{i=1}^m E\big \vert Y_i\big\vert^3}{ (\sum_{i=1}^m \sigma_i^2)^{3/2}} \lesssim \frac{1}{\sigma^3} \sum_{i=1}^m E\big \vert Y_i\big\vert^3,
\end{split}
\end{equation}
where $C$ is a constant, $F_Z(x):=P\left(Z\leq x\right)$, and $ \Phi(x) := P\left(N(0,1)\leq x\right)$.

Note that 
\begin{equation}
\begin{split}
|Y_i| = \begin{cases}
\frac{(1-p)\eta(e_i)}{pl}\text{ with prob. } p 
 \\
\frac{\eta(e_i)}{l} \text{ with prob. }1-p,
\end{cases}
\end{split}
\end{equation}
then $E|Y_i|^3 = \frac{(1-p)((1-p)^2 + p^2)}{p^2l^3} \eta^3(e_i)$. Thus, 
\begin{equation}
\begin{split}
\frac{1}{\sigma^3} \sum_{i=1}^m E\big \vert Y_i\big\vert^3 
&= \frac{ \frac{(1-p)((1-p)^2 + p^2)}{p^2l^3}}{ (\frac{1-p}{pl^2})^{3/2}} \cdot
 \frac{\sum_{i=1}^m\eta^3(e_i)} {(\sum_{i=1}^m\eta^2(e_i))^{3/2}} \\
 & = \frac{(1-p)^2 + p^2}{\sqrt{p(1-p)}}
 \cdot
 \frac{\sum_{i=1}^m\eta^3(e_i)} {(\sum_{i=1}^m\eta^2(e_i))^{3/2}}
 \\ 
 &\lesssim  \frac{\sum_{i=1}^m\eta^3(e_i)} {(\sum_{i=1}^m\eta^2(e_i))^{3/2}}.
 \end{split}
\end{equation}
Therefore, if $\frac{\sum_{i=1}^m\eta^3(e_i)} {(\sum_{i=1}^m\eta^2(e_i))^{3/2}} \rightarrow 0$, as $m \rightarrow \infty$, we have 
\begin{equation}
\begin{split}
    \sup_{x\in \mathbb{R}} |F_Z(x) - \Phi(x)| \lesssim \frac{1}{\sigma^3} \sum_{i=1}^m E\big \vert Y_i\big\vert^3 \rightarrow 0, \text{ as } m \rightarrow \infty, 
\end{split}
\end{equation}
thus $Z \xrightarrow{D} N(0,1)$, as $m \rightarrow \infty$. 
\end{proof}

\section{Proof of Lemma \ref{lem-1}}\label{appendix_proof_lemma}
Under Assumption \ref{3} that $\eta(e_i) \leq C_2 < \infty$ for sufficiently large $i$, where $C_2$ is a constant, and $C(H, T_m) \gg \sqrt{m}$, we can show as follows that assumption \eqref{1} holds,
\begin{equation}
 \frac{\sum_{i = 1}^m \eta^2(e_i)}{(\sum_{i = 1}^m \eta(e_i))^2} = 
\frac{\sum_{i = 1}^m \eta^2(e_i)}{l^2 C^2(H, T_m)} 
\lesssim \frac{m}{C^2(H, T_m)} \cdot \frac{C_2^2}{l^2} \rightarrow 0 \text{ as }   m \rightarrow \infty.
\end{equation}
Hence, Assumption \ref{3} implies assumption \eqref{1}. The rest of Lemma \ref{lem-1} can be proved similarly. 

\section{Proof of Proposition \ref{prop:ci}} \label{appendix_proof_prop}
Note that 
\begin{equation}
    E[ \hat \sigma^2(H, T_m)] = \frac{1-p}{p^2_m l^2}  \sum_{i = 1}^m E[\omega_i] \eta^2(e_i) =  \sigma^2(H, T_m),
\end{equation}
\begin{equation}
\begin{split}
        var[ \hat \sigma^2(H, T_m)/\sigma^2(H, T_m)] &= \frac{1}{p^2_m(\sum_{i = 1}^m  \eta^2(e_i))^2 }  \sum_{i = 1}^m var[\omega_i] \eta^4(e_i)  \\
        & = \frac{1-p}{p} \frac{ \sum_{i = 1}^m \eta^4(e_i)}{(\sum_{i = 1}^m  \eta^2(e_i))^2}.
\end{split}
\end{equation}
Then by Chebyshev's inequality, under Assumption \ref{4}, for every $\epsilon > 0$
\begin{equation}
\begin{split}
P\left(\big | \frac{\hat  \sigma(H, T_m)}{\sigma(H, T_m)} - 1\big| > \epsilon \right) 
& \leq \frac{var[ \hat \sigma^2(H, T_m)/\sigma^2(H, T_m)]}{\epsilon^2}  \\
& = \frac{1-p}{p\epsilon^2} \frac{ \sum_{i = 1}^m \eta^4(e_i)}{(\sum_{i = 1}^m  \eta^2(e_i))^2} \\
&\lesssim
\frac{1-p}{p\epsilon^2} \frac{m}{C^2(H, T_m)}\rightarrow  0, 
 \text{ as }m \rightarrow \infty
\end{split}
\end{equation}
Hence $\frac{\hat \sigma^2(H, T_m)}{ \sigma^2(H, T_m)} \rightarrow_P 1 $. \eqref{equ:CI} is an immediate consequence of this consistency and $ Z(H, T_m) \rightarrow_D N(0, 1) $, i.e. 
\begin{equation}
    \frac{\hat C(H, T_m) - C(H, T_m)}{\hat \sigma(H, T_m)} = Z(H, T_m) \cdot \frac{\hat \sigma^2(H, T_m)}{ \sigma^2(H, T_m)} \rightarrow_D N(0, 1).
\end{equation}

\section{Proof of Theorem \ref{the:clt-st}}\label{appendix:proof_clt-st} 
Note that $T(\tau)$ is a marked Poisson point process, where arrival times $t_i$ follow a Poisson process and each is marked by an edge $e_i$. Also note that the random variables $\eta(e_i)$ are defined as functions of this process over windows centered at each $t_i$, and these windows may be overlapping. Thus, in general, the $\eta(e_i)$'s are dependent and $ \hat C(H, T(\tau)) = \frac{1}{p l } \sum_{i = 1}^{N(\tau)} \omega_i\eta(e_i)$ is therefore not a compound Poisson process. The exception is the case
when $l=1$, where we have $\eta(e_i) = 1$, so that the count estimator becomes $\hat C(H, T(\tau)) = \frac{1}{p l } \sum_{i = 1}^{N(\tau)} \omega_i$. Consistency and the stated CLT hold in this case using standard results for compound Poisson processes, since the $\omega_i$ are i.i.d., independent of $N(\tau)$.  Our goal now is to show that the stated results continue to hold for $l\geq 2$, under the two conditions of the theorem w.r.t. the parameters $\lambda, \tau$, and the total number of vertices $|V|$ in the stochastic model of $T(\tau)$.

To prove consistency for $\hat  C(H, T(\tau))$, we need to show that for every  $\epsilon>0$, 
\begin{equation}
    P(| \frac{\hat  C(H, T(\tau))}{ C(H, T(\tau))} - 1|>\epsilon) \rightarrow 0, \text{ as } \tau \rightarrow \infty.
\end{equation}

Let $\Omega$ denote the sample space on which $T(\tau)$ is defined, and $\mathcal{F}$ the $\sigma$-algebra on $\Omega$. 
For any $\epsilon_0>0$ and $C_0>0$, define $A \in \mathcal{F}$:
\begin{equation}
    A := \{T(\tau) \in \Omega:\frac{N(\tau)^{2/3}}{C(H,T(\tau))}<\epsilon_0, \text{ and } \eta(e_i)\leq C_0 \text{ for } i=1,\cdots, N(\tau) \}.
\label{equ:A_set}
\end{equation}
Then we have a partition: $\Omega = A \cup A^c$, where $A^c$ is the complement of set $A$. 

For every $\epsilon>0$,
\begin{equation}
\begin{split}
        P\left(| \frac{\hat  C(H, T(\tau))}{ C(H, T(\tau))} - 1|>\epsilon \right)
        & =   E \left[ \bm {I}_{\{| \frac{\hat  C(H, T(\tau))}{ C(H, T(\tau))} - 1|>\epsilon\}}\right] 
     = S_1 + S_2,
\end{split}
\end{equation}
where based on the law of total expectation,
\begin{equation}
\begin{split}
        S_1 :& =  E \left[ \bm {I}_{\{| \frac{\hat  C(H, T(\tau))}{ C(H, T(\tau))} - 1|>\epsilon\}} \,| \, A \right]   P(A) \\
        &=  E\left[   E\left[ \bm {I}_{\{| \frac{\hat  C(H, T(\tau))}{ C(H, T(\tau))} - 1|>\epsilon\}} \,| \, A, T(\tau)=T_m \right] \right] P(A),\\
    S_2 : &=   E \left[ \bm {I}_{\{| \frac{\hat  C(H, T(\tau))}{ C(H, T(\tau))} - 1|>\epsilon\}} \,| \, A^c \right]   P(A^c),
\end{split}
\end{equation}
For the first expectation $S_1$, the inner expectation is taken with respect to the sampling distribution of $\hat  C(H, T(\tau))$ conditional on $T(\tau)$ being a deterministic $T_m \in A$. The outer expectation is taken with respect to the distribution of $T(\tau)$ conditional on $T(\tau)\in A$. Hence, conditional on any deterministic $T_m \in A$, from the consistency theorem for deterministic graphs $T_m$, it follows that,
\begin{equation}
\begin{split}
E\left[ \bm {I}_{\{| \frac{\hat  C(H, T(\tau))}{ C(H, T(\tau))} - 1|>\epsilon\}} \,| \, A, T(\tau)=T_m \right] & = 
P \left( | \frac{\hat  C(H, T(\tau))}{ C(H, T(\tau))} - 1| > \epsilon   \,|\,  A, T(\tau)  = T_m \right)  \\
&\leq  \frac{1-p}{\epsilon^2 p } \frac{\sum_{i = 1}^m \eta^2(e_i)}{(\sum_{i = 1}^m \eta(e_i))^2}   \\
     & \leq   \frac{(1-p)C_0^2}{\epsilon^2l^2 p }\cdot   (\frac{m^{1/2}}{C(H, T_m)})^2 \\
     & \leq \frac{(1-p)C_0^2}{\epsilon^2l^2 p }\cdot   (\frac{m^{2/3}}{C(H, T_m)})^2\\
     & < \frac{(1-p)C_0^2}{\epsilon^2l^2 p }\cdot \epsilon_0^2.
\end{split}
\end{equation}
Hence, 
\begin{equation}
\begin{split}
   S_1  \leq  E\left[\frac{(1-p)C_0^2}{\epsilon^2l^2 p }\cdot \epsilon_0^2 \right ]  P(A)  \leq  \frac{(1-p)C_0^2}{\epsilon^2l^2 p }\cdot \epsilon_0^2.
\label{equ:s1}
\end{split}
\end{equation}
Then for every $\epsilon$, fixed sampling ratio $p$, the sum term $S_1$ can be made arbitrarily small by making $\epsilon_0$ small in the construction of set $A$. That is, if we set $\epsilon_0 \rightarrow 0$, as $\tau\rightarrow \infty$ in \eqref{equ:A_set}, then for every $\epsilon >0$, 
\begin{equation}
    S_1 \rightarrow 0, \text{ as } \tau\rightarrow \infty.
\end{equation}

For the second expectation $S_2$, since $ E \left[ \bm {I}_{\{| \frac{\hat  C(H, T(\tau))}{ C(H, T(\tau))} - 1|>\epsilon\}} \,| \, A^c \right] \leq 1$, then 
\begin{equation}
\begin{split}
      S_2 \leq  P(A^c) = 1-P\left(A\right)
      \rightarrow 0, \text{ as } \tau \rightarrow \infty,
\end{split}
\end{equation}
if $P\left(A\right) \rightarrow 1$, as $\tau \rightarrow \infty$. Therefore, in order to show consistency in the stochastic case for $T(\tau)$, it suffices to show 
\begin{equation}
   P\left(A\right) \rightarrow 1 , \text{ as } \tau \rightarrow \infty,
\label{equ:assumption_to_prove}
\end{equation}
where in the construction of set $A$, $\epsilon_0 \rightarrow 0$, as $\tau\rightarrow \infty$. The proof is provided in the supplemental materials. 

We now show that the CLT also follows if \eqref{equ:assumption_to_prove} holds.
To prove a CLT for $\hat  C(H, T(\tau))$, we need show that for every $x\in \mathbb{R}$,
\begin{equation}
    |P\left(Z(H, T(\tau))\leq x\right)-\Phi(x) | \rightarrow 0, \text{ as } \tau \rightarrow \infty,
\end{equation}
where $ \Phi(x) := P\left(N(0,1)\leq x\right)$. To simplify the notation, we drop the dependency on $H$ and $T(\tau)$ from $Z(H, T(\tau))$, and denote it by $Z(\tau)$, when no confusion is possible.

From the law of total expectation, we have
\begin{equation}
\begin{split}
     P\left(Z(\tau)\leq x\right) 
      &= 
   E \left[ \bm {I}_{\{Z(\tau)\leq x\}}\right]  \\
   &= E \left[ \bm {I}_{\{Z(\tau)\leq x\}}\,\Big |\, A\right] P(A) + E \left[  \bm {I}_{\{Z(\tau)\leq x\}}
   \,\Big |\, A^c\right] P(A^c),
\end{split}
\end{equation}
and hence for each fixed $x$,  
\begin{equation}
\begin{split}
      \big|
     P\left(Z(\tau)\leq x\right)-\Phi(x) 
     \big| 
      &= 
     \big |  P\left(Z(\tau)\leq x\right) - \Phi(x)P(A) - \Phi(x) (1-P(A))
     \big| \\
     & \leq S_3 + S_4,
\end{split}
\end{equation}
where
\begin{equation}
    \begin{split}
        S_3 &:=  
        \big|
        E \left[ \bm {I}_{\{Z(\tau)\leq x\}}\,\Big |\, A\right]  - \Phi(x)\big| \cdot P(A) \\
      &  =\big|  E\left[E \left[ \bm {I}_{\{Z(\tau)\leq x\}} 
     \,\Big |\, A, T(\tau)=T_m\right] \right]-
      \Phi(x)\big| \cdot P(A)  \\
      &  =\big|  E\left[E \left[ \bm {I}_{\{Z(\tau)\leq x\}} 
     \,\Big |\, A, T(\tau)=T_m\right] -
      \Phi(x)\right]\big| \cdot P(A)  \\
      &  \leq  E\left[\big|  E \left[ \bm {I}_{\{Z(\tau)\leq x\}} 
     \,\Big |\, A, T(\tau)=T_m\right] -
      \Phi(x)\big|\right] \cdot P(A),  \\
    S_4 &:=  \big|
        E \left[ \bm {I}_{\{Z(\tau)\leq x\}}\,\Big |\, A^c\right]  - \Phi(x)\big| \cdot P(A^c)  .
    \end{split}
\end{equation}
Conditional on $T(\tau)$ being a deterministic $T_m \in A$, from the CLT theorem for deterministic graphs $T_m$, it follows that
\begin{equation}
\begin{split}
 \big | E \left[ \bm {I}_{\{Z(\tau)\leq x\}} \,\Big |\, A, T(\tau)=T_m\right] -\Phi(x) \big|
 &=
   \big |P \left( Z(\tau)\leq x \,|\,   A, T(\tau) = T_m \right)  - \Phi(x)\big|
   \\& \leq  \frac{(1-p)^2 + p^2}{\sqrt{p(1-p)}}
 \cdot
 \frac{\sum_{i=1}^m\eta^3(e_i)} {(\sum_{i=1}^m\eta^2(e_i))^{3/2}}  \\
     & \leq   
      \frac{C_0^3((1-p)^2 + p^2)}{l^{3/2}\sqrt{p(1-p)}} \cdot   (\frac{m^{2/3}}{C(H, T_m)})^{3/2} \\
     & <  \frac{C_0^3((1-p)^2 + p^2)}{l^{3/2}\sqrt{p(1-p)}}  \cdot \epsilon_0^{3/2}.
\end{split}
\end{equation}
Similarly as in \eqref{equ:s1}, for fixed sampling ratio $p$, the term $S_3$ can be made arbitrarily small by making $\epsilon_0$ small. That is, if we set $\epsilon_0 \rightarrow 0$, as $\tau\rightarrow \infty$ in \eqref{equ:A_set}, then 
\begin{equation}
    S_3 \rightarrow 0, \text{ as } \tau\rightarrow \infty.
\end{equation}
For the term $S_4$, since $ \big|
        E \left[ \bm {I}_{\{Z(\tau)\leq x\}}\,\Big |\, A^c\right]  - \Phi(x)\big|\leq 2$, then 
\begin{equation}
\begin{split}
      S_4 \leq   2 \cdot  (1-P\left(A\right))
      \rightarrow 0, \text{ as } \tau \rightarrow \infty,
\end{split}
\end{equation}
if $P\left(A\right) \rightarrow 1$, as $\tau \rightarrow \infty$. Therefore, in order to show CLT in the stochastic case for $T(\tau)$, it again suffices to show \eqref{equ:assumption_to_prove} holds. The proof of \eqref{equ:assumption_to_prove} is provided in the supplemental materials..

\section*{Acknowledgments}
This work was supported in part by NSF award SES-2120115.
\bibliographystyle{siamplain}
\bibliography{ms}
\end{document}


\maketitle

\section{Proof of (E.9)}\label{sec:proof_assumption}

Write 
\begin{equation}
     A = A_1 \cap A_2,
\end{equation}
where
\begin{equation}
\begin{split}
   A_1 &:= \{T(\tau) \in \Omega:\eta(e_i)\leq C_0 \text{ for } i=1,\cdots, N(\tau) \}, \\
   A_2 & :=  \{T(\tau) \in \Omega:\frac{N(\tau)^{2/3}}{C(H,T(\tau))}<\epsilon_0\}.
\end{split}
\end{equation}
To show 
(E.9) holds, it suffices to show 
\begin{equation}
\begin{split}
      &P(T(\tau) \in  A_1) \rightarrow 1 , \text{ as } \tau \rightarrow \infty, \\
      &P(T(\tau) \in  A_2) \rightarrow 1 , \text{ as } \tau \rightarrow \infty.
\end{split}
\end{equation}

These tasks boil down to showing the following under the two
conditions of the theorem:
\begin{enumerate}
    \item 
    for each $i=1,\cdots, N(\tau)$, the random local counts $\eta(e_i) = O_p(1)$ as $\tau \rightarrow \infty$, (i.e., $\eta(e_i) \leq C_0 < \infty$ with high probability for $\tau$ large enough, where $C_0$ is a constant),
    \item 
\begin{equation}
 \frac{N(\tau)^{2/3}}{C(H,T(\tau))}   \xrightarrow{p} 0,  \text{ as } \tau \rightarrow \infty.
\end{equation}    
    
\end{enumerate}

We first consider the random local count $\eta(e_i)$, for $i =1,...,N(\tau)$, which is the number of instances of an $l$-edge $\delta$-temporal motif $H$ in $T(\tau)$ containing an edge $e_i$. Let $\kappa$ be the number of temporal edges that occur within a time interval of length $2\delta$ centered at the (random) time $t_i$ when $e_i$ occurs. Then there are at most $\kappa \choose l$ edge subsequences of length $l$ occurring within $\delta$ of $e_i$ that could possibly form an instance of motif $H$. Thus, $\eta(e_i) \leq {\kappa \choose l}$.  
Since we assume that the underlying Poisson process is stationary, it follows that for each $\eta(e_i)$, $\kappa$ is zero-truncated Poisson distributed (reflecting that $e_i$ is known to occur within this $2\delta$ interval) with mean $\frac{2\lambda\delta}{1-\exp(-2\lambda\delta)}$.

\noindent\textit{Remark}: 
Let $X \sim ZTPoisson(\lambda)$, and $Y \sim Poisson(\lambda)$, then for $k=1,2,\cdots$, 
$    P(X=k) = \frac{1}{1-\exp(-\lambda)}P(Y=k)
$, and $E(X^k) = \frac{1}{1-\exp(-\lambda)} E(Y^k)$.

We will use Markov's inequality to prove our result.  Note that $\eta(e_i) \leq { \kappa \choose l} \leq (\frac{e}{l})^l \kappa^l$. Using a bound for moments of the Poisson distribution (Theorem 1 in \cite{ahle2021sharp}), we have 
\begin{equation}
\begin{split}
    E \left[\eta(e_i) \right] 
    &\leq  (\frac{e}{l})^l  \cdot E \left[\kappa^l \right] \\
    & \leq 
    (\frac{e}{l})^l  \cdot 
    \frac{1}{1-\exp(-2\lambda\delta)} \cdot
    \left(\frac{l}
    {\log(l/(2\lambda\delta) + 1)} \right)^l  <\infty \text{ if } \lambda < \infty.
\end{split}
\end{equation}
Denote the upper bound on the right-hand side above by $A(\lambda)$. Hence, the expected local count $E \left[\eta(e_i) \right]$ is bounded above by $A(\lambda)$ for each edge $e_i$ occurring in $(0, \tau]$, when $\lambda < \infty$. By analyzing this function, we know that $A(\lambda)$ is increasing w.r.t. $\lambda$, $A(\lambda)$ is finite whenever $\lambda \neq \infty$, and $\lim_{\lambda\rightarrow 0}A(\lambda)=0$. 

By Markov's inequality, $\forall t > 0$, we have
\begin{equation}
    P(\eta(e_i) \geq t) \leq \frac{E \left[\eta(e_i) \right] }{ t}.
\end{equation}
Under condition 1 of the Theorem 
3.4, $\lambda$ is bounded above, say by $\lambda_0\in (0, \infty)$, as $\tau \rightarrow \infty$, so $A(\lambda)$ is also bounded above by $A(\lambda_0)$ for sufficiently large $\tau$. Hence, for every $\epsilon > 0$, there exists constant $C_\epsilon =A(\lambda_0)/\epsilon$ and $\tau_0$ such that for all $\tau > \tau_0$,
\begin{equation}\label{local_2}
    P(\eta(e_i) \geq C_\epsilon) \leq \frac{E \left[\eta(e_i) \right] }{ A(\lambda_0)/\epsilon} \leq \frac{A(\lambda) }{ A(\lambda_0)/\epsilon} \leq \epsilon.
\end{equation}
Thus, $\eta(e_i) = O_p(1)$, as $\tau \rightarrow \infty $. 

We now show that 
\begin{equation}
 \frac{N(\tau)^{2/3}}{C(H,T(\tau))}   \xrightarrow{p} 0,  \text{ as } \tau \rightarrow \infty.
\end{equation}    
Noting that 
\begin{equation}
   \frac{N(\tau)^{2/3}}{C(H,T(\tau))}  = 
   \frac{N(\tau)^{2/3}}{E\left[N(\tau)^{2/3} \right]}  
   \cdot 
   \frac{E\left[N(\tau)^{2/3} \right]}{E\left[C(H,T(\tau))\right]} 
   \cdot 
   \frac{E\left[C(H,T(\tau))\right]}{C(H,T(\tau))},
\end{equation}
our proof proceeds by establishing the following three results: (1) $\frac{N(\tau)^{2/3}}{E\left[N(\tau)^{2/3} \right]}  \xrightarrow{p} 1$, as $\tau \rightarrow \infty$ (Lemma \ref{lemma_proof_1}), (2)   $\frac{E\left[N(\tau)^{2/3}\right]}{E\left[C(H,T(\tau))\right]} \rightarrow 0$, as $\tau \rightarrow \infty$ (Lemma \ref{lemma_proof_2}), and (3) $\frac{E\left[C(H,T(\tau))\right]}{C(H,T(\tau))} = O_p(1)
$ (Lemma \ref{lemma_proof_3}).
The result then follows by Slutsky's theorem and the rule that $O_p(1)o_p(1)=o_p(1)$.



We begin with the following lemma.
\begin{lemma}\label{lemma_proof_1}
Assume that $\lambda\tau \rightarrow \infty$, as $\tau \rightarrow \infty$, then we have 
\begin{equation}
    \frac{N(\tau)^{2/3}}{E\left[N(\tau)^{2/3} \right]}  \xrightarrow{p} 1, \text{ as } \tau \rightarrow \infty
\end{equation}
\end{lemma}

\begin{proof}
Recall that $N(\tau) \sim Poisson(\lambda \tau)$. A Taylor expansion for $N(\tau)^{2/3}$ about the mean $E\left[N(\tau) \right] =\lambda\tau$ yields that in expectation 
\begin{equation}
\begin{split}
& E\left[ N(\tau)^{2/3} \right] \\
& = 
(\lambda\tau)^{2/3} + \frac{2}{3} (\lambda\tau)^{-1/3}E\left[( N(\tau) - \lambda\tau) \right]
    - \frac{1}{9}(\lambda\tau)^{-4/3} E\left[( N(\tau) - \lambda\tau)^2\right]+\cdots \\
    & = (\lambda\tau)^{2/3} + o((\lambda\tau)^{-1/3}).
\end{split}
\end{equation}
Thus, 
\begin{equation}\label{E1}
   \big(E\left[ N(\tau)^{2/3} \right] \big)^2 = (\lambda\tau)^{4/3} + \Theta((\lambda\tau)^{1/3})+o((\lambda\tau)^{-2/3}).
\end{equation}
Similarly, Taylor expansion for $N(\tau)^{4/3}$ yields
\begin{equation}\label{E2}
    E\left[  N(\tau)^{4/3}\right] = (\lambda\tau)^{4/3} + \Theta((\lambda\tau)^{1/3}) + o((\lambda\tau)^{-2/3})
\end{equation}
Combining \eqref{E1} and \eqref{E2}, we have  
\begin{equation}
\begin{split}
        Var\left[ N(\tau)^{2/3}\right] &= E\left[  N(\tau)^{4/3}\right] - \big(E\left[ N(\tau)^{2/3} \right] \big)^2 \\
        & = \Theta((\lambda\tau)^{1/3}) + o((\lambda\tau)^{-2/3})
\end{split}
\end{equation}

By Chebyshev's inequality, we have for every $\epsilon > 0$, 
\begin{equation}
\begin{split}
    & P \left(\big|  \frac{N(\tau)^{2/3}}{E\left[N(\tau)^{2/3} \right]} - 1\big|> \epsilon \right) 
    = P\left( \big | N(\tau)^{2/3} - E[N(\tau)^{2/3}] \big |> E[N(\tau)^{2/3} ] \epsilon \right ) \\
    & \leq \frac{Var[N(\tau)^{2/3}]}{(E[N(\tau)^{2/3}])^2\epsilon^2}\\
    & = 
    \frac{1}{\epsilon^2}  \frac{\Theta((\lambda\tau)^{1/3}) + o((\lambda\tau)^{-2/3})}{(\lambda\tau)^{4/3} + \Theta((\lambda\tau)^{1/3})+o((\lambda\tau)^{-2/3})} \rightarrow 0, \text{ as } \lambda\tau \rightarrow \infty. 
\end{split}
\end{equation}
Since $\lambda\tau \rightarrow \infty$, as $\tau \rightarrow \infty$ (this holds under condition 1), the lemma is proved.
\end{proof}


\begin{lemma}\label{lemma_proof_2}
Assume that $\lambda\tau \rightarrow \infty$, as $\tau \rightarrow \infty$, then we have 
\begin{equation}
    \frac{E\left[N(\tau)^{2/3}\right]}{E\left[C(H,T(\tau))\right]} \rightarrow 0, \text{ as } \tau \rightarrow \infty.
\end{equation}
\end{lemma}

\begin{proof}
Given that $N(\tau) \sim Poisson(\lambda \tau)$, using a recent upper bound for moments of the Poisson distribution (Theorem 1 in \cite{ahle2021sharp}), we have 
\begin{equation}\label{local_1}
   E\left[N(\tau)^{2/3}\right]\leq 
     \left(\frac{2/3}
    {\log( 2/(3\lambda\tau) + 1)} \right)^{2/3} \rightarrow \infty, \text{ as } \tau \rightarrow \infty.
\end{equation}

We now establish a lower bound for $E\left[C(H,T(\tau))\right]$ and show that it grows faster than the right-hand side of (\ref{local_1}). Noting that 
\begin{equation}
    E\left[C(H,T(\tau))\right] = E\left[ E\left[C(H,T(\tau)) \,\big | \, N(\tau)\right]\right],
\end{equation}
we first consider $E\left[C(H,T(\tau)) \,\big | \, N(\tau)=m \right]$, the expected count of $k$-node, $l$-edge $\delta$-temporal motif $H$ in $T(\tau)$, conditional on $N(\tau) = m$. Hereafter, we will write the random count variable $C(H, T(\tau))$ as $C(H, T_m)$ when conditional on $N(\tau)=m$.

Recall the property of Poisson processes that conditional on $N(\tau) = m$, the $m$ (unordered) arrival times are independent and identically distributed (i.i.d.) with a $Unif(0,\tau)$ distribution. Let $U_1,U_2,\cdots,U_m$ denote the $m$ unordered arrival times, and $e_{u_1}, e_{u_2}, \cdots, e_{u_m}$ denote the corresponding edges occurring at each of these arrival times. Under the current specification of our random graph model that edges at each arrival times are sampled uniformly among the set of all possible edges, independent of each other, then $e_{u_i}, i=1,\cdots,m$ are independent, each of which is distributed uniformly on the state space $\mathcal{E}=\{1,2,3,\cdots,|V|(|V|-1)\}$ (assume we consider directed edges), where $|V|$ is the total number of vertices in the graph.

Let $X_i = (U_i, e_{u_i})$ be a random vector, taking values in $\mathcal{X} = [0,\tau] \times \mathcal{E}$, then $X_i,i=1,\cdots m $ are independent and identically distributed conditional on $N(\tau)=m$, and the two components of each $X_i$ are also independent. We can see that conditional on $N(\tau) =m$, $\{X_1,\cdots,X_m\}$ is an unordered list of timestamped edges in $T_m =\{ (e_i, t_i)$, $i = 1, \cdots, m\}$, which carries the same information as $T_m$. 

Then conditional on $N(\tau)=m$, the random count $C(H, T_m)$ can be written as a sum of Bernoulli random variables, 
\begin{equation}\label{count_con}
    C(H,T_m)= \sum_{i_1<i_2<\cdots<i_l} h(X_{i_1}, X_{i_2},\cdots, X_{i_l}),
\end{equation}
where the sum is taken over all $m\choose l$ combinations of distinct indices $1\leq i_1<i_2<\cdots<i_l\leq m$, and the function $h: \mathcal{X}^l \rightarrow \{0,1\}$ is defined as
\begin{equation}\label{h_kernel}
\begin{split}
     & h(x_1,x_2,\cdots, x_l) \\
     &= h((u_1,e_{u_1}),(u_2,e_{u_2}),\cdots, (u_l,e_{u_1})) \\
    &= 
    I\{range(u_1, u_2,\dots,u_l) \leq \delta\}   \cdot
    I\{ e_{u_{(1)}},e_{u_{(2)}}, \cdots,e_{u_{(l)}} \text{ is isomorphic to }H
    \},
\end{split}
\end{equation}
where $u_{(1)}<u_{(2)}<\cdots<u_{(l)}$ are the order statistics of $u_1,u_2,\cdots, u_l$, $range(u_1,u_2,\cdots, u_l) = u_{(l)} - u_{(1)}$, and we say a sequence of edges is isomorphic to temporal motif $H$ if it matches the same edge pattern as $H$ and all of the edges occur in the right order regardless of time duration. 

Hence, by applying function $h$ to $l$ distinct time-stamped edges in $T_m$, we're checking whether the time ordered sequence of the $l$ edges forms an instance of the $\delta$-temporal motif $H$, and return $1$ if it does, otherwise $0$. We then obtain the total motif count by applying $h$ to all $m\choose l$ combinations of distinct time-stamped edges in $T_m$, and sum up all returned values from function $h$. 

Since $X_i,i=1,\cdots,m$ are i.i.d. conditional on $N(\tau)=m$, then taking conditional expectation on both sides in \eqref{count_con}, we have
\begin{equation}\label{eq-28}
    E\left[C(H,T_m)\,|\,N(\tau)=m \right] = {m\choose l} E\left[ h(X_1,X_2,\cdots,X_l)\right].
\end{equation}
By the independence of the two components $U_i$ and $e_{u_i}$ in $X_i$, we have 
\begin{equation}\label{eq-29}
\begin{split}
     E\left[ h(X_1,X_2,\cdots,X_l)\right]
    &= P\left(range(U_1, U_2,\dots,U_l) \leq \delta \right) \cdot \\
    & \quad \quad \quad 
    P\left( e_{u_{(1)}},e_{u_{(2)}}, \cdots,e_{u_{(l)}} \text{ is isomorphic to }H\right) \\
    & = \pi_{\delta, l,\tau} \cdot C_{|V|,k},
\end{split}
\end{equation}
where
\begin{equation}
    \begin{split}
        C_{|V|,k} &:= P\left( e_{u_{(1)}},e_{u_{(2)}}, \cdots,e_{u_{(l)}} \text{ is isomorphic to }H\right)\\
        &={|V|\choose k}k! (\frac{1}{|V|(|V|-1)})^l\asymp |V|^{-(2l-k)}, \text{ for sufficient large $|V|$ w.r.t. $k$},
    \end{split}
\end{equation}
and 
\begin{equation}
    \begin{split}
\pi_{\delta, l,\tau} &:=  P\left(range(u_1, u_2,\dots,u_l) \leq \delta  \right)\\
& = \idotsint_{B_l} f(u_1,\dots,u_l) \,du_1 \dots du_l  
\\
& = \idotsint_{B_l} \frac{1}{\tau^l} \,du_1 \dots du_l \\
 &=  \idotsint_{\tilde{B}_l} 1 \,dz_1 \dots dz_l, \quad(\text{obtained by setting } z_i=u_i/\tau).
    \end{split}
\end{equation}
where $f(\cdot)$ is the joint p.d.f of $U_1,\cdots,U_l \stackrel{i.i.d.}{\sim} Unif(0,\tau)$, and $B_l$ and $\tilde{B}_l$ are defined as follows,
\begin{equation}
\begin{split}
        B_l &= \{(u_1,u_2,\cdots,u_l) \in [0,\tau]^l: range(u_1,\cdots,u_l)\leq \delta\}, \\
\tilde{B}_l &= \{(z_1,z_2,\cdots,z_l) \in [0,1]^l: range(z_1,\cdots,z_l)\leq \frac{\delta}{\tau}\}.
\end{split}
\end{equation} 

Hence, $ \pi_{\delta, l,\tau}$ is the volume of the subspace $\tilde{B}_l$ inside the $l$-dimensional cube, then $ \pi_{\delta, l,\tau}$ is decreasing w.r.t. $\tau$, and $\pi_{\delta, l,\tau} \rightarrow 0$, as $\tau \rightarrow \infty$.

When $l=2$, $\tilde{B}_l$ is the shaded area as shown in Figure \ref{fig-integral}, we can calculate that for sufficiently large $\tau$, 
\begin{equation}
    \pi_{\delta, l,\tau} = vol(\tilde{B}_l) = 1- (1-\frac{\delta}{\tau})^2=\frac{2\delta}{\tau} - \frac{\delta^2}{\tau^2} \asymp \frac{1}{\tau}.
\end{equation}

\begin{figure}[H]
    \centering
    \includegraphics[scale=0.2]{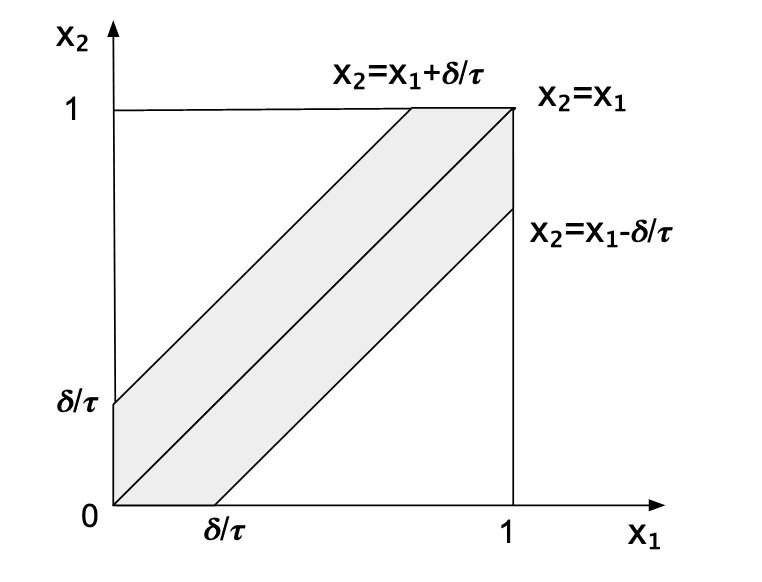}
\caption{When $l=2$, $\tilde{B}_l$ is shown as the shaded area.}
    \label{fig-integral}
\end{figure}

For $l\geq 2$, we provide the following lower bound.
\begin{claim}\label{claim}
$\pi_{\delta,l,\tau} \gtrsim \frac{1}{\tau^{l-1}}, l\geq 2.$
\end{claim}
To see so, define  
\begin{equation}
A_{l,k} = \left\{(z_1,z_2,\cdots,z_l) \in [0,1]^l: z_1, z_2,\cdots, z_l \in \left[(k-1) \frac{\delta}{\tau}, k \frac{\delta}{\tau}\right]\right\},
\end{equation}
for $k=1,2,\cdots,  \lfloor \frac{\tau}{\delta} \rfloor$, formed by the small $l$-dimensional hyper-cubes with length $\delta/\tau$ located along the main diagonal line of the large $l$-dimensional hyper-cube with length $1$. See Figure \ref{fig2} for illustration in 2- and 3-dimensional space. 

We can see that every point $(z_1,\cdots, z_l) \in A_{l,k} $ satisfies the condition $range(z_1,\cdots, z_l) \leq \delta/\tau$, and there are at most $\lfloor \frac{1}{\delta/\tau} \rfloor = \lfloor \frac{\tau}{\delta} \rfloor$ mutually disjoint hyper-cubes with length $\delta/\tau$ along the diagonal line of the unit hyper-cube, hence,
\begin{equation}
   \bigcup_{k=1}^{  \lfloor \frac{\tau}{\delta} \rfloor} A_{l,k} \subseteq \tilde{B}_l.
\end{equation}
Therefore, 
\begin{equation}
    \pi_{\delta,l,\tau} = vol(\tilde{B}_l) \geq \sum_{k=1}^{ \lfloor \frac{\tau}{\delta} \rfloor} vol(A_{l,k} ) = 
     \lfloor \frac{\tau}{\delta} \rfloor (\frac{\delta}{\tau})^l \asymp \frac{1}{\tau^{l-1}}.
\end{equation}
This completes the proof of Claim \ref{claim}. 

\begin{figure}[H]
    \centering
  \includegraphics[scale=0.20]{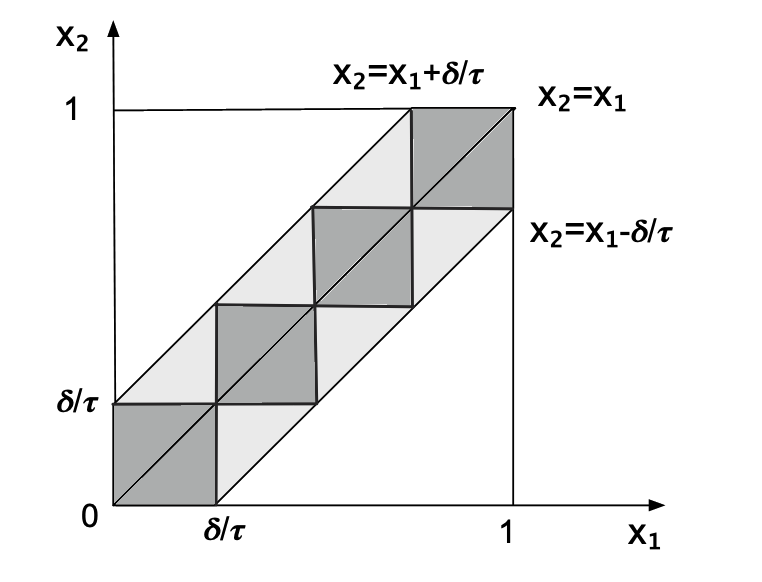}
    \includegraphics[scale=0.17]{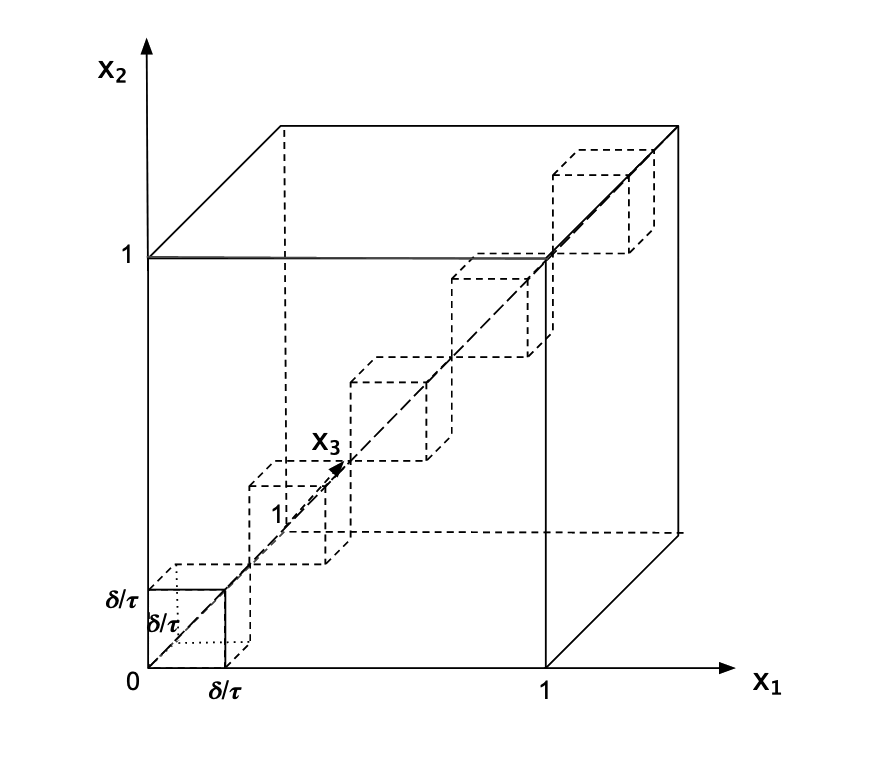}
    \caption{Illustration of $A_{l,k}$ in 2- and 3-dimensional space. Left: 2-dimensional squares with length $\delta/\tau$ located along main diagonal line of the unit square; right:  3-dimensional cubes with length $\delta/\tau$ located along main diagonal line of the unit cube.}
    \label{fig2}
\end{figure}

To complete our proof of the lemma, we note one fact of binomial coefficient $m\choose l$ is that $(\frac{1}{l})^l m^l \leq {m \choose l}\leq (\frac{e}{l})^l m^l$, so ${m\choose l} \asymp m^l$ when $m$ is sufficiently large compared with $l$. 
Combining $\eqref{eq-28}$ and $\eqref{eq-29}$, and defining ${m\choose l} := 0$ when $m<l$, we have for sufficiently large $\tau$,
\begin{equation}\label{E_c_same}
\begin{split}
      E\left[C(H,T(\tau))\right] &=  E\left[E\left[C(H,T_m)\,|\,N(\tau)=m \right] \right] \\
      &= E[{N(\tau)\choose l}] \cdot
  \pi_{\delta, l,\tau}C_{|V|,k} \\
  & = \pi_{\delta, l,\tau}C_{|V|,k} \sum_{m=0}^\infty {m \choose l } e^{-\lambda\tau} \frac{(\lambda\tau)^m}{m!} \\
  & \asymp \pi_{\delta, l,\tau}C_{|V|,k} \sum_{m=0}^\infty {m^l } e^{-\lambda\tau} \frac{(\lambda\tau)^m}{m!} \\
  & = \pi_{\delta, l,\tau}C_{|V|,k} 
  E \left[ N(\tau)^l\right] \enskip .
\end{split}
\end{equation}
Using a recent lower bound for the moments of a Poisson distribution (Section 2.2 Equation 9 in \cite{ahle2021sharp}), and claim \ref{claim}, we have 
\begin{equation}\label{E_c_lower}
     E\left[C(H,T(\tau))\right] \gtrsim 
     C_{|V|,k} \cdot 
     \frac{1}{\tau^{l-1}} \cdot 
     (\lambda\tau)^l(1+\frac{l(l-1)}{2\lambda\tau}) \gtrsim  C_{|V|,k} \cdot \lambda^l\tau
\end{equation}
Therefore, combining \eqref{local_1} and \eqref{E_c_lower}, we have 
\begin{equation}
\begin{split}
\frac{ E\left[N(\tau)^{2/3}\right]}
{ E\left[C(H,T(\tau))\right] }
  & \lesssim
  \frac{
       \left(\frac{2/3}
    {\log( 2/(3\lambda\tau) + 1)} \right)^{2/3}
  }
  {
  \lambda^l\tau
  }
  \cdot
  \frac{1}
  { C_{|V|,k} } \\
  & =        \left(\frac{2/3}
    {\log( 2/(3\lambda\tau) + 1) \cdot (\lambda^l\tau)^{3/2}} \right)^{2/3} \cdot
  \frac{1}
  { C_{|V|,k} }
  \rightarrow 0, \text{ as } \tau \rightarrow \infty.
\end{split}
\end{equation}
To see so, we apply L'Hôpital's rule on the denominator and obtain  \begin{equation}
    \lim_{\tau\rightarrow \infty} \frac{\log( 2/(3\lambda\tau) + 1)}{(\lambda^l\tau)^{-3/2}} =  \lim_{\tau\rightarrow \infty} \frac{4}{9} \cdot
    \frac{1}{\frac{2}{3\lambda{\tau}}+1} \cdot \tau^{1/2} \cdot \lambda^{\frac{3}{2}l-1} = \infty.
\end{equation}

\end{proof}


\begin{lemma}\label{lemma_proof_3}
Assume that $\lambda\tau \rightarrow \infty$, as $\tau \rightarrow \infty$, then we have 
\begin{equation}
    \frac{E\left[C(H,T(\tau))\right]}{C(H,T(\tau))} = O_p(1), \text{ as } \tau\rightarrow \infty.
\end{equation}
\end{lemma}

\begin{proof}
Write 
\begin{equation}
\begin{split}
    \frac{E\left[C(H,T(\tau))\right]}{C(H,T(\tau))} = \frac{E\left[C(H,T(\tau))\right]/E[N(\tau)^l]}{C(H,T(\tau))/N(\tau)^l} \cdot \frac{E[N(\tau)^l]}{N(\tau)^l}.
\end{split}
\end{equation}
Note that we have shown in \eqref{E_c_same} that for sufficiently large $\tau$, 
\begin{equation}
     \frac{E\left[C(H,T(\tau))\right]}{ E \left[ N(\tau)^l\right]}  \asymp \pi_{\delta, l,\tau}C_{|V|,k}.
\end{equation}
And similarly to the proof for Lemma \eqref{lemma_proof_2}, we can show through Taylor expansion that 
\begin{equation}
\begin{split}
 E\left[  N(\tau)^{l}\right] & = (\lambda\delta)^l + o((\lambda\delta)^l), \\
 E\left[  N(\tau)^{2l}\right] & = (\lambda\delta)^{2l} + o((\lambda\delta)^{2l}), \\
        Var\left[ N(\tau)^l\right] &= E\left[  N(\tau)^{2l}\right] - \big(E\left[ N(\tau)^{l} \right] \big)^2 =o((\lambda\delta)^{2l})\, .
\end{split}
\end{equation}
Then by Chebyshev's inequality, we have for every $\epsilon > 0$, 
\begin{equation}
\begin{split}
    P \left(\big|  \frac{N(\tau)^{l}}{E\left[N(\tau)^{l} \right]} - 1\big|>\epsilon \right) 
    & \leq \frac{Var[N(\tau)^{l}]}{(E[N(\tau)^{l}])^2\epsilon^2}\\
    & = 
    \frac{1}{\epsilon^2}  \frac{o((\lambda\delta)^{2l})}
    {(\lambda\delta)^{2l} + o((\lambda\delta)^{2l})} \rightarrow 0, \text{ as } \lambda\tau \rightarrow \infty, 
\end{split}
\end{equation} 
hence $\frac{N(\tau)^l}{E[N(\tau)^l]}  \xrightarrow{p} 1$, as $\tau\rightarrow \infty$. Thus by the continuous mapping theorem, we have
\begin{equation}
\frac{E[N(\tau)^l]}{N(\tau)^l}  \xrightarrow{p} 1, \text{ as } \tau\rightarrow \infty.
\end{equation}
Therefore, in order to show $\frac{E\left[C(H,T(\tau))\right]}{C(H,T(\tau))} = O_p(1)
$, it suffices to show $    \frac{\pi_{\delta, l,\tau}C_{|V|,k}}{C(H,T(\tau))/N(\tau)^l} = O_p(1)$.

From \eqref{count_con}, we can see that $C(H, T_m)/{m\choose l}$ is a U-statistic with expectation $E[C(H, T_m)/{m\choose l}]=\pi_{\delta, l,\tau}C_{|V|,k} < \infty$. By the law of large numbers for U-statistics \cite{hoeffding1961strong}, we have $C(H, T_m)/{m\choose l} \xrightarrow{p} \pi_{\delta, l,\tau}C_{|V|,k}$, as $m\rightarrow \infty$. Since ${m \choose l} \asymp m^l$, as $m\rightarrow \infty$, then 
\begin{equation}
    C(H, T_m)/m^l \xrightarrow{p} \pi_{\delta, l,\tau}C_{|V|,k}, \text{ as } m\rightarrow \infty \enskip .
\end{equation}
Thus for an arbitrary $\gamma \in (0,1)$, we have 
\begin{equation}\label{eq:C_m_asym}
\begin{split}
    &P\left(C(H,T_m)/m^l <\gamma \cdot\pi_{\delta, l,\tau}C_{|V|,k} \right) \\
    &=     P\left(C(H,T_m)/m^l - \pi_{\delta, l,\tau}C_{|V|,k} <-(1-\gamma)\pi_{\delta, l,\tau}C_{|V|,k} \right) \\
    &\leq
    P\left(\,|\,C(H,T_m)/m^l - E[C(H,T_m)/m^l]\,|\, > (1-\gamma)\pi_{\delta, l,\tau}C_{|V|,k} \right) \\
    & \rightarrow 0, \text{ as } m \rightarrow \infty.
\end{split}
\end{equation}
Also note that for any fixed $m$,  
\begin{equation}
    P(N(\tau) = m ) = \frac{(\lambda \tau)^m}{m!\exp(\lambda \tau)} \rightarrow 0, \text{ as } \tau \rightarrow \infty.
\label{eq:C_tau_asym}
\end{equation}

Combining \eqref{eq:C_m_asym} and \eqref{eq:C_tau_asym}, we have that 
\begin{equation}
\begin{split}
     & P \left(
     \frac{\pi_{\delta, l,\tau}C_{|V|,k}}{C(H,T(\tau))/N(\tau)^l} > 1/\gamma \right)\\
     &= P\left(C(H,T(\tau))/N(\tau)^l<\gamma \cdot\pi_{\delta, l,\tau}C_{|V|,k}\right) \\
   & = \sum_{m=0}^\infty P\left(C(H,T_m)/m^l <\gamma \cdot\pi_{\delta, l,\tau}C_{|V|,k} \right) \cdot P\left(N(\tau)=m\right) \rightarrow 0, \text{ as } \tau \rightarrow \infty.
\end{split}
\end{equation}
Thus, $\frac{\pi_{\delta, l,\tau}C_{|V|,k}}{C(H,T(\tau))/N(\tau)^l}$ is bounded in probability, and so is $\frac{E\left[C(H,T(\tau))\right]}{C(H,T(\tau))}$.

\end{proof}

\bibliographystyle{siamplain}
\bibliography{ms}

%% file: shared.tex
\usepackage{lipsum}
\usepackage{amsfonts}
\usepackage{graphicx}
\usepackage{epstopdf}
\usepackage{algorithmic}
\ifpdf
  \DeclareGraphicsExtensions{.eps,.pdf,.png,.jpg}
\else
  \DeclareGraphicsExtensions{.eps}
\fi
\usepackage{enumitem}   

\newsiamremark{remark}{Remark}
\newsiamremark{hypothesis}{Hypothesis}
\crefname{hypothesis}{Hypothesis}{Hypotheses}
\newsiamthm{claim}{Claim}
\newtheorem{assumption}{Assumption}
\newtheorem{prop}{Proposition}
\usepackage{amssymb}
\usepackage{bm}
\headers{Quantifying Uncertainty for Temporal Motif Estimation}
{Xiaojing Zhu, and Eric D. Kolaczyk}

\title{Quantifying Uncertainty for Temporal Motif Estimation in Graph Streams under Sampling
}

\author{Xiaojing Zhu\footnotemark[1] 
\and Eric D. Kolaczyk\thanks{Department of Mathematics and Statistics, Boston University, Boston, MA 02215 USA (\email{xiaojzhu@bu.edu}, \email{kolaczyk@bu.edu}).}}

\usepackage{amsopn}


%% file: ms.bbl
\begin{thebibliography}{10}

\bibitem{aggarwal2011outlier}
{\sc C.~C. Aggarwal, Y.~Zhao, and S.~Y. Philip}, {\em Outlier detection in
  graph streams}, in 2011 IEEE 27th international conference on data
  engineering, IEEE, 2011, pp.~399--409.

\bibitem{ahmed2021online}
{\sc N.~K. Ahmed, N.~Duffield, and R.~A. Rossi}, {\em Online sampling of
  temporal networks}, ACM Transactions on Knowledge Discovery from Data (TKDD),
  15 (2021), pp.~1--27.

\bibitem{ahmed2017sampling}
{\sc N.~K. Ahmed, N.~Duffield, T.~Willke, and R.~A. Rossi}, {\em On sampling
  from massive graph streams}, arXiv preprint arXiv:1703.02625,  (2017).

\bibitem{battiston2017multilayer}
{\sc F.~Battiston, V.~Nicosia, M.~Chavez, and V.~Latora}, {\em Multilayer motif
  analysis of brain networks}, Chaos: An Interdisciplinary Journal of Nonlinear
  Science, 27 (2017), p.~047404.

\bibitem{bhattacharya2020motif}
{\sc B.~B. Bhattacharya, S.~Das, and S.~Mukherjee}, {\em Motif estimation via
  subgraph sampling: The fourth moment phenomenon}, arXiv preprint
  arXiv:2011.03026,  (2020).

\bibitem{chechik2008activity}
{\sc G.~Chechik, E.~Oh, O.~Rando, J.~Weissman, A.~Regev, and D.~Koller}, {\em
  Activity motifs reveal principles of timing in transcriptional control of the
  yeast metabolic network}, Nature biotechnology, 26 (2008), pp.~1251--1259.

\bibitem{chen2017unified}
{\sc X.~Chen and J.~C. Lui}, {\em A unified framework to estimate global and
  local graphlet counts for streaming graphs}, in Proceedings of the 2017
  IEEE/ACM International Conference on Advances in Social Networks Analysis and
  Mining 2017, 2017, pp.~131--138.

\bibitem{eswaran2018sedanspot}
{\sc D.~Eswaran and C.~Faloutsos}, {\em Sedanspot: Detecting anomalies in edge
  streams}, in 2018 IEEE International Conference on Data Mining (ICDM), IEEE,
  2018, pp.~953--958.

\bibitem{gomes2019machine}
{\sc H.~M. Gomes, J.~Read, A.~Bifet, J.~P. Barddal, and J.~Gama}, {\em Machine
  learning for streaming data: state of the art, challenges, and
  opportunities}, ACM SIGKDD Explorations Newsletter, 21 (2019), pp.~6--22.

\bibitem{gut2005probability}
{\sc A.~Gut and A.~Gut}, {\em Probability: a graduate course}, vol.~200,
  Springer, 2005.

\bibitem{holme2015modern}
{\sc P.~Holme}, {\em Modern temporal network theory: a colloquium}, The
  European Physical Journal B, 88 (2015), pp.~1--30.

\bibitem{holme2019temporal}
{\sc P.~Holme and J.~Saram{\"a}ki}, {\em Temporal network theory}, vol.~2,
  Springer, 2019.

\bibitem{hulovatyy2015exploring}
{\sc Y.~Hulovatyy, H.~Chen, and T.~Milenkovi{\'c}}, {\em Exploring the
  structure and function of temporal networks with dynamic graphlets},
  Bioinformatics, 31 (2015), pp.~i171--i180.

\bibitem{jazayeri2020motif}
{\sc A.~Jazayeri and C.~C. Yang}, {\em Motif discovery algorithms in static and
  temporal networks: A survey}, Journal of Complex Networks, 8 (2020),
  p.~cnaa031.

\bibitem{kashtan2004efficient}
{\sc N.~Kashtan, S.~Itzkovitz, R.~Milo, and U.~Alon}, {\em Efficient sampling
  algorithm for estimating subgraph concentrations and detecting network
  motifs}, Bioinformatics, 20 (2004), pp.~1746--1758.

\bibitem{klusowski2018counting}
{\sc J.~M. Klusowski and Y.~Wu}, {\em Counting motifs with graph sampling}, in
  Conference On Learning Theory, PMLR, 2018, pp.~1966--2011.

\bibitem{kovanen2011temporal}
{\sc L.~Kovanen, M.~Karsai, K.~Kaski, J.~Kert{\'e}sz, and J.~Saram{\"a}ki},
  {\em Temporal motifs in time-dependent networks}, Journal of Statistical
  Mechanics: Theory and Experiment, 2011 (2011), p.~P11005.

\bibitem{kovanen2013temporal}
{\sc L.~Kovanen, K.~Kaski, J.~Kert{\'e}sz, and J.~Saram{\"a}ki}, {\em Temporal
  motifs reveal homophily, gender-specific patterns, and group talk in call
  sequences}, Proceedings of the National Academy of Sciences, 110 (2013),
  pp.~18070--18075.

\bibitem{liu2018sampling}
{\sc P.~Liu, A.~Benson, and M.~Charikar}, {\em A sampling framework for
  counting temporal motifs}, arXiv preprint arXiv:1810.00980,  (2018).

\bibitem{liu2021temporal}
{\sc P.~Liu, V.~Guarrasi, and A.~E. Sariyuce}, {\em Temporal network motifs:
  Models, limitations, evaluation}, IEEE Transactions on Knowledge and Data
  Engineering,  (2021).

\bibitem{mackey2018chronological}
{\sc P.~Mackey, K.~Porterfield, E.~Fitzhenry, S.~Choudhury, and G.~Chin}, {\em
  A chronological edge-driven approach to temporal subgraph isomorphism}, in
  2018 IEEE International Conference on Big Data (Big Data), IEEE, 2018,
  pp.~3972--3979.

\bibitem{matias2018semiparametric}
{\sc C.~Matias, T.~Rebafka, and F.~Villers}, {\em A semiparametric extension of
  the stochastic block model for longitudinal networks}, Biometrika, 105
  (2018), pp.~665--680.

\bibitem{milo2002network}
{\sc R.~Milo, S.~Shen-Orr, S.~Itzkovitz, N.~Kashtan, D.~Chklovskii, and
  U.~Alon}, {\em Network motifs: simple building blocks of complex networks},
  Science, 298 (2002), pp.~824--827.

\bibitem{min2013burstiness}
{\sc B.~Min and K.-I. Goh}, {\em Burstiness: Measures, models, and dynamic
  consequences}, in Temporal networks, Springer, 2013, pp.~41--64.

\bibitem{panzarasa2009patterns}
{\sc P.~Panzarasa, T.~Opsahl, and K.~M. Carley}, {\em Patterns and dynamics of
  users' behavior and interaction: Network analysis of an online community},
  Journal of the American Society for Information Science and Technology, 60
  (2009), pp.~911--932.

\bibitem{paranjape2017motifs}
{\sc A.~Paranjape, A.~R. Benson, and J.~Leskovec}, {\em Motifs in temporal
  networks}, in Proceedings of the Tenth ACM International Conference on Web
  Search and Data Mining, 2017, pp.~601--610.

\bibitem{perry2013point}
{\sc P.~O. Perry and P.~J. Wolfe}, {\em Point process modelling for directed
  interaction networks}, Journal of the Royal Statistical Society: Series B
  (Statistical Methodology), 75 (2013), pp.~821--849.

\bibitem{prvzulj2007biological}
{\sc N.~Pr{\v{z}}ulj}, {\em Biological network comparison using graphlet degree
  distribution}, Bioinformatics, 23 (2007), pp.~e177--e183.

\bibitem{ramirez2017survey}
{\sc S.~Ram{\'\i}rez-Gallego, B.~Krawczyk, S.~Garc{\'\i}a, M.~Wo{\'z}niak, and
  F.~Herrera}, {\em A survey on data preprocessing for data stream mining:
  Current status and future directions}, Neurocomputing, 239 (2017),
  pp.~39--57.

\bibitem{rocha2017sampling}
{\sc L.~E. Rocha, N.~Masuda, and P.~Holme}, {\em Sampling of temporal networks:
  Methods and biases}, Physical Review E, 96 (2017), p.~052302.

\bibitem{sanei2019fleet}
{\sc S.-V. Sanei-Mehri, Y.~Zhang, A.~E. Sariy{\"u}ce, and S.~Tirthapura}, {\em
  Fleet: Butterfly estimation from a bipartite graph stream}, in Proceedings of
  the 28th ACM International Conference on Information and Knowledge
  Management, 2019, pp.~1201--1210.

\bibitem{sarpe2021presto}
{\sc I.~Sarpe and F.~Vandin}, {\em Presto: Simple and scalable sampling
  techniques for the rigorous approximation of temporal motif counts}, in
  Proceedings of the 2021 SIAM International Conference on Data Mining (SDM),
  SIAM, 2021, pp.~145--153.

\bibitem{shin2017wrs}
{\sc K.~Shin}, {\em Wrs: Waiting room sampling for accurate triangle counting
  in real graph streams}, in 2017 IEEE International Conference on Data Mining
  (ICDM), IEEE, 2017, pp.~1087--1092.

\bibitem{song2014event}
{\sc C.~Song, T.~Ge, C.~Chen, and J.~Wang}, {\em Event pattern matching over
  graph streams}, Proceedings of the VLDB Endowment, 8 (2014), pp.~413--424.

\bibitem{tu2018network}
{\sc K.~Tu, J.~Li, D.~Towsley, D.~Braines, and L.~D. Turner}, {\em Network
  classification in temporal networks using motifs}, arXiv preprint
  arXiv:1807.03733,  (2018).

\bibitem{wang2020efficient}
{\sc J.~Wang, Y.~Wang, W.~Jiang, Y.~Li, and K.-L. Tan}, {\em Efficient sampling
  algorithms for approximate temporal motif counting}, in Proceedings of the
  29th ACM International Conference on Information \& Knowledge Management,
  2020, pp.~1505--1514.

\bibitem{wei2017identifying}
{\sc Y.~Wei, X.~Liao, C.~Yan, Y.~He, and M.~Xia}, {\em Identifying topological
  motif patterns of human brain functional networks}, Human brain mapping, 38
  (2017), pp.~2734--2750.

\bibitem{zignani2017temporal}
{\sc M.~Zignani, C.~Quadri, M.~Del~Vicario, S.~Gaito, and G.~P. Rossi}, {\em
  Temporal communication motifs in mobile cohesive groups}, in International
  Conference on Complex Networks and their Applications, Springer, 2017,
  pp.~490--501.

\end{thebibliography}


\begin{thebibliography}{1}

\bibitem{ahle2021sharp}
{\sc T.~D. Ahle}, {\em Sharp and simple bounds for the raw moments of the
  binomial and poisson distributions}, arXiv preprint arXiv:2103.17027,
  (2021).

\bibitem{hoeffding1961strong}
{\sc W.~Hoeffding}, {\em The strong law of large numbers for u-statistics.},
  tech. report, North Carolina State University. Dept. of Statistics, 1961.

\end{thebibliography}
